\documentclass[twocolumn,english,footinbib]{revtex4-1}
\usepackage[T1]{fontenc}
\usepackage[latin9]{inputenc}
\usepackage{babel}
\usepackage{verbatim}
\usepackage{float}
\usepackage{units}
\usepackage{amsmath}
\usepackage{amsthm}
\usepackage{amssymb}
\usepackage{graphicx}
\usepackage{esint}
\usepackage{xcolor}
\usepackage[unicode=true,pdfusetitle,
 bookmarks=true,bookmarksnumbered=false,bookmarksopen=false,
 breaklinks=false,pdfborder={0 0 1},backref=false,colorlinks=false]
 {hyperref}
 \usepackage{hyperref}

\makeatletter

\providecommand{\tabularnewline}{\\}

\theoremstyle{plain}
\newtheorem{thm}{\protect\theoremname}
  \theoremstyle{definition}
  \newtheorem{defn}[thm]{\protect\definitionname}
  \theoremstyle{plain}
  \newtheorem{cor}[thm]{\protect\corollaryname}
  \theoremstyle{plain}
  \newtheorem{lem}[thm]{\protect\lemmaname}
  \theoremstyle{plain}
  \newtheorem{prop}[thm]{\protect\propositionname}
  \theoremstyle{remark}
  \newtheorem{rem}[thm]{\protect\remarkname}

\makeatother

  \providecommand{\corollaryname}{Corollary}
  \providecommand{\definitionname}{Definition}
  \providecommand{\lemmaname}{Lemma}
  \providecommand{\propositionname}{Proposition}
  \providecommand{\remarkname}{Remark}
\providecommand{\theoremname}{Theorem}

\begin{document}

\bibliographystyle{unsrt}

\title{Continuous-variable entanglement distillation over a pure loss channel
\\ with multiple quantum scissors}

\author{Kaushik P. Seshadreesan$^{1}$, Hari Krovi$^{2}$, and Saikat Guha$^{1}$}

\affiliation{$^{1}$College of Optical Sciences, University of Arizona, Tucson,
AZ 85721, USA\\$^{2}$Quantum Engineering and Computing Physical Sciences
and Systems, Raytheon BBN Technologies, Cambridge, MA 02138, USA}

\date{\today}

\begin{abstract}
Entanglement distillation is a key primitive for distributing high quality entanglement between remote locations. Probabilistic noiseless linear amplification based on the quantum scissors is a candidate for entanglement distillation from noisy continuous-variable (CV) entangled states. Being a non-Gaussian operation, the quantum scissors is challenging to analyze. We present a derivation of the non-Gaussian state heralded by multiple quantum scissors in a pure loss channel with two-mode squeezed vacuum input. We choose the reverse coherent information (RCI)---a proven lower bound on the distillable entanglement of a quantum state under one-way local operations and classical communication (LOCC), as our figure of merit. We evaluate a Gaussian lower bound on the RCI of the heralded state. We show that it can exceed the unlimited two-way LOCC-assisted direct transmission entanglement distillation capacity of the pure loss channel. The optimal heralded Gaussian RCI with two quantum scissors is found to be significantly more than that with a single quantum scissors, albeit at the cost of decreased success probability. Our results fortify the possibility of a quantum repeater scheme for CV quantum states using the quantum scissors.
\end{abstract}

\maketitle

\section{Introduction}
Entanglement shared across large distances is a key resource for quantum information processing tasks such as distributed quantum computation~\cite{RvM16,DDKP07,BR03}, distributed sensing~\cite{PKD18,GJEGF18,ZZS18}, quantum communication
protocols such as quantum key distribution~\cite{E91}, quantum teleportation~\cite{BBCJPW93}, superdense coding~\cite{BW92}, entanglement assisted high rate quantum error correcting codes~\cite{GJG18}, and entanglement-assisted classical communication over noisy channels~\cite{BSST99}. 

Optical photons are arguably the best carriers of quantum information to distribute entanglement between remote locations~\cite{Kimble2008}. Optical entanglement distribution is largely classified into schemes based on discrete and continuous variables (DV and CV) depending on the entangled resource state that is transmitted~\cite{PCLWZ12,HHHH09}. In DV, the resource states are maximally entangled states of discrete, finite-dimensional degrees of freedom of single photons, such as the polarization; while in CV, they are entangled multimode squeezed states of the continuous, infinite-dimensional, quadrature degrees of freedom of electromagnetic field modes, which are Gaussian states, i.e., states completely described by the first two moments of the quadrature operators associated with the modes.

The primary challenge in distributing optical entanglement is photon loss and noise in transmission, which degrades the quality of entanglement. Thus, entanglement distillation---the process of distilling from several copies of a noisy entangled state fewer copies of a more entangled state using local operations and classical communication (LOCC), is key in distributing high quality entanglement. Entanglement distillation in DV, e.g., from weakly entangled mixed states of pairs of single-photon polarization qubits, typically involves the quantum CNOT gate~\cite{BBPSSW96, DEJMPS96}, or more practically, a simple polarizing beam splitter~\cite{PSBZ01}, acting locally on pairs of identical copies of the state, followed by measurement and classical communication between the two parties. When applied recursively, they yield highly entangled and highly pure two-qubit states~\cite{DB07}. In CV, it is known that entanglement distillation from noisy Gaussian entangled states cannot be effected by Gaussian operations alone~\cite{ESP02, NGGL14}, where the latter refer to state transformations based on Hamiltonians that are at most quadratic in the quadrature operators and map Gaussian states to other Gaussian states. Non-Gaussian operations such as photon counting are required. In this regard, schemes based on Fock state filtering~\cite{SRZ06}, such as quantum catalysis~\cite{UFPKRL15, LM02}, photon subtraction~\cite{TNTTHFS10, DZNLFPW12}, symmetric photon replacement~\cite{BESP03,EBSP04,LR09} and purifying distillation~\cite{Furasek10}, which degaussify Gaussian states, have been considered. When applied recursively, these schemes regaussify the final output while yielding highly entangled~\cite{BESP03, DZNLFPW12}, and in some cases also highly pure Gaussian states~\cite{LR09, Furasek10}.

In Ref.~\cite{RL09}, Ralph and Lund proposed the concept of probabilistic noiseless linear amplification (NLA), which can be realized using the non-Gaussian operation---the quantum scissors~\cite{PPB98}, as a candidate for CV entanglement distillation. In the limit of a large number of scissors, NLA can probabilistically distill highly entangled and highly pure Gaussian states from weakly entangled, mixed Gaussian states~\cite{RL09}. Approximate NLA based on a single quantum scissors was further investigated in~\cite{DR2017,DR2018} for quantum error correction~\cite{Ralph2011} of CV entangled states, towards designing quantum repeaters~\cite{MATN15,SSdRG11,BDCZ98}. CV entanglement distillation with a single quantum scissors was experimentally demonstrated in~\cite{XRLWP10}.

In this paper, we analyze entanglement distillation with approximate NLA effected by multiple, but finite number of quantum scissors over a pure loss channel with two-mode squeezed vacuum state (TMSV) input. We show that the entanglement content of the heralded output of the NLA---quantified by the {\em reverse coherent information} (RCI)~\cite{GPLS09,PGBL09,DJKR06,Hcube00}---is higher than 
	\begin{equation}
	C_{\rm direct}(\eta) = -\log(1-\eta),
	\label{plob}
	\end{equation}
the recently proven maximum entanglement generation rate in ebits (maximally entangled qubits) per mode achievable through a pure-loss channel of transmissivity $\eta$~\cite{PLOB17,PGBL09} (see also~\cite{WTB17} for a strong converse theorem.) 
The RCI is an information-theoretic lower bound on the distillable entanglement per copy of a shared state that is achievable using one-way LOCC when many copies of the shared state are available~\cite{PLOB17,DW05}. Although NLA has been proposed in the past for CV entanglement distillation, this is the first time the optimal trade-off between the RCI of the heralded state and the probability of success of the NLA as a function of the scissor-based NLA's internal parameters, a gain (transmissivity of a beamsplitter) and the number of scissors is quantified. Clearly, RCI times the NLA success probability must be less than $C_{\rm direct}(\eta)$~\cite{PLOB17}. But, being able to herald entanglement over a lossy channel, even if probabilistically, of RCI higher than $C_{\rm direct}(\eta)$, i.e., heralding a state of distillable entanglement higher than $C_{\rm direct}(\eta)$, is significant towards realizing CV quantum repeaters since the latter is a pre-requisite to building a second-generation quantum repeater that can outperform $C_{\rm direct}(\eta)$~\cite{SKG18}.

More technically, the contributions of this paper include a calculation of the non-Gaussian state heralded by the approximate NLA based on multiple but finite number of quantum scissors and the corresponding heralding success probability, based on characteristic functions and the Husimi Q function. The calculation also applies to the teleportation-based CV error correction scheme with approximate NLA~\cite{Ralph2011}, and is computationally efficient. This is in contrast to the Fock basis calculations presented in~\cite{DR2017,DR2018}, which do not scale well with increasing number of quantum scissors. Our choice of the RCI as the figure of merit is operationally more relevant than those considered before, such as the logarithmic negativity~\cite{VW02, Plenio05} and the entanglement of formation (EOF)~\cite{BDSW96}. We numerically evaluate a Gaussian lower bound on the RCI of the state heralded by approximate NLA. We show that there exist TMSV input mean photon numbers and scissors gain parameter values where the lower bound exceeds $C_{\rm direct}(\eta)$. Further, we find that the Gaussian RCI heralded by two quantum scissors is significantly higher compared to that heralded by a single quantum scissors, and in some cases the addition of the second quantum scissors even helps boost it above $C_{\rm direct}(\eta)$, which could not be achieved with just a single quantum scissors. In addition to the analyses based on the RCI, we also evaluate a Gaussian lower bound on the EOF of the state heralded by a single quantum scissors. We qualitatively validate the findings presented in~\cite{DR2018} and extend the analysis to two quantum scissors.

The paper is organized as follows. In Sec.~\ref{sec:Noiseless-Linear-Amplification},
we review the basic concept of NLA with quantum scissors, outlining its relevance to CV entanglement distillation. In Sec.~\ref{sec:Entanglement-distillation-over}, we describe in detail the methods we employ to analyze CV entanglement distillation using multiple, but finite number of quantum scissors, including our figure of merit the RCI. Section~\ref{sec:Results} contains the results. Finally, we summarize our findings in Sec.~\ref{sec:Conclusions}.%

\section{Noiseless Linear Amplification with Quantum Scissors\label{sec:Noiseless-Linear-Amplification}}

Noiseless linear amplification~\cite{RL09} refers to probabilistic amplification that, e.g., transforms coherent states as $\left|\alpha\right\rangle \rightarrow\left|g\alpha\right\rangle$
where $g\in\mathbb{R}$ is the gain. NLA can be implemented in a heralded fashion using linear optics,
photon injection and detection~\cite{RL09}. The scheme involves splitting
the input signal into $N$ parts of equal intensity and recombining
them following the quantum scissors operation on each part
as shown in Fig.~\ref{LOPD NLA} (a). 

\begin{figure}[H]
	\centering
	
	\begin{tabular}{c}
		\includegraphics[scale=0.475]{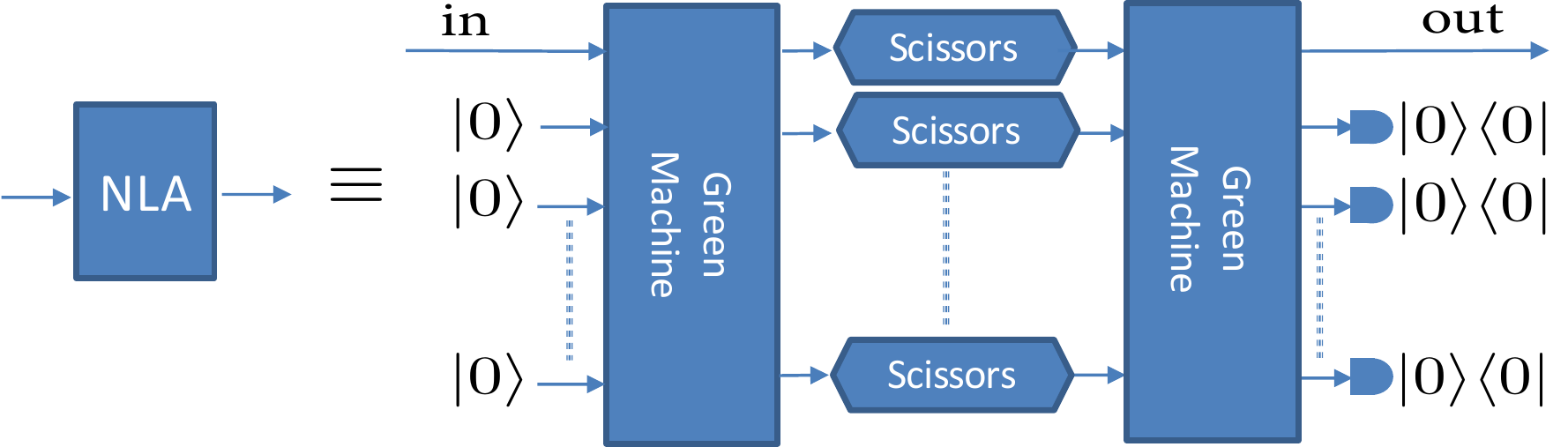}\tabularnewline
	\end{tabular}
	
	(a)
	
	\medskip{}
	
	\begin{tabular}{c}
		\includegraphics[scale=0.55]{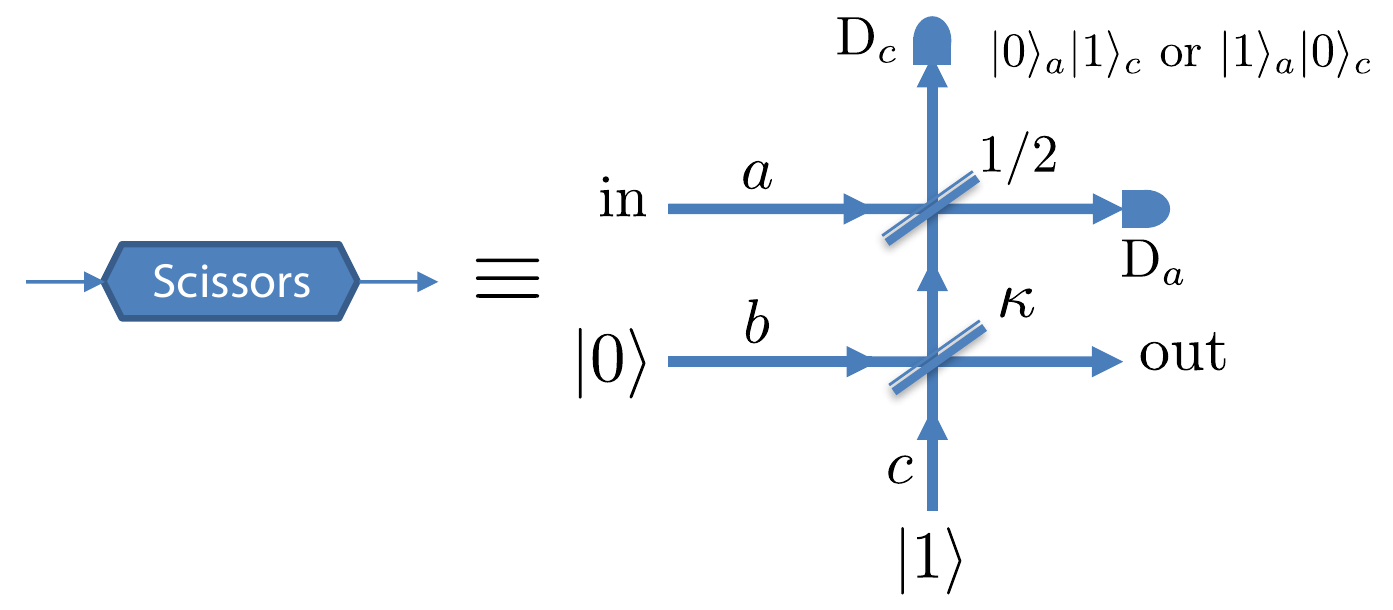}\tabularnewline
	\end{tabular}
	
	(b)
	
	\caption{(a) Noiseless linear amplification (NLA) implemented using linear
		optics and (b) the quantum scissors operation based on photon
		injection and detection. The Green Machines in (a) refer to an $n-$mode input $n-$mode device that performs a $n-$way splitting or its inverse (recombining) unitary operation. When all but one of the input modes are in the vacuum state, the Green machine splits the mean photon number in the first input mode uniformly across all the output modes.}
	\label{LOPD NLA}
\end{figure}
Quantum scissors, as the name suggests, refers to an operation that
truncates a quantum state in Fock space~\cite{PPB98}. In addition to truncation, it can be used to amplify
certain Fock state components of the state relative to others~\cite{RL09}.
Consider the scheme shown in Fig.~\ref{LOPD NLA} (b) comprising of
single photon injection and detection. A single photon (in mode $c$)
is mixed with vacuum (in mode $b$) on a beam splitter of transmissivity
$\kappa=1/\left(1+g^{2}\right)$ ($g$ being the intended gain of
NLA), creating an entangled state in the $\left\{ \left|0\right\rangle _{b}\otimes\left|1\right\rangle _{c},\left|1\right\rangle _{c}\otimes\left|0\right\rangle _{b}\right\} $
subspace. When the signal in mode $a$ is mixed with mode $c$ on
a 50:50 beam splitter and either one of the two projections $\left\{ \left|0\right\rangle _{a}\otimes\left|1\right\rangle _{c},\left|1\right\rangle _{a}\otimes\left|0\right\rangle _{c}\right\} $
is applied (i.e., when detector $D_{a}$ clicks and $D_{c}$ doesn't,
or vice versa), the $\left\{ \left|0\right\rangle ,\left|1\right\rangle \right\} $
support of the quantum state of the signal is teleported to mode $b.$
Further, the teleported state is such that its $\left|1\right\rangle $
component is amplified relative to the vacuum component depending
on the choice of $\kappa$. In summary, the quantum scissors
scheme of Fig.~\ref{LOPD NLA} on any input signal state $\left|\psi\right\rangle \propto\alpha_{0}\left|0\right\rangle +\alpha_{1}\left|1\right\rangle +\cdots$
heralds the truncated and amplified output
\begin{equation}
\Gamma\left(g\right)\left(\alpha_{0}\left|0\right\rangle +\alpha_{1}\left|1\right\rangle +\cdots\right)=\sqrt{\frac{{1}}{1+g^{2}}}\left(\alpha_{0}\left|0\right\rangle +g\alpha_{1}\left|1\right\rangle \right).\label{eq: scissor op}
\end{equation}

\begin{figure}[H]
	\centering
	
	\begin{tabular}{c}
		\includegraphics[scale=0.65]{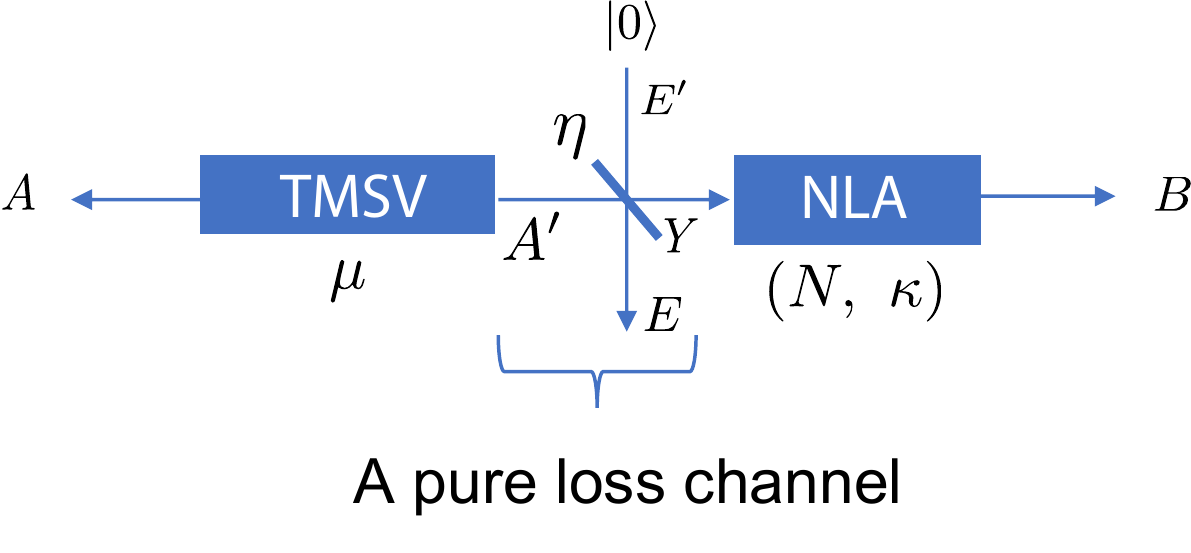}\tabularnewline
	\end{tabular}
	
	\caption{A pure loss channel of transmissivity $\eta$ with TMSV state input
		appended by $N-$quantum scissors of gain $g=\sqrt{\left(1-\kappa\right)/\kappa}$
		at the output.}
	\label{half channel}
\end{figure}

When the signal to the modified quantum scissors is sufficiently weak
such that its quantum state resides primarily in the $\left\{ \left|0\right\rangle ,\left|1\right\rangle \right\} $
subspace, the operation effects NLA; whereas if the state has significant
support on higher photon components, then the amplification is not
noiseless anymore owing to the excess noise originating from the truncation
of the teleported state in Fock space. Thus, in the scheme of Fig.~\ref{LOPD NLA} (a), for a given input signal intensity, $N$ needs
to be sufficiently large so that the sub-signals that are inputs to
the quantum scissors operations are weak. When all the quantum scissors
operations succeed and all-but-one of the outputs of the $N-$combiner
are measured in the vacuum state, the device approaches NLA.

As mentioned earlier, NLA is probabilistic. The success probability
of NLA with $N-$quantum scissors is input state dependent
and decreases exponentially with $N.$ For an input coherent state
$\left|\alpha\right\rangle ,$ the success probability of an
$N-$scissors NLA drops with $N$ and the NLA gain $g$ as $P_{s}=1/\left(1+g^{2}\right)^{N}e^{-\left(1-g^{2}\right)\left|\alpha\right|^{2}}$~\cite{RL09}.

NLA is particularly relevant to CV entanglement distillation. Consider a TMSV state of mean photon number $\mu=\sinh^2 r$ 
\begin{align}
|\Psi\rangle_{AA'}=\sqrt{1-\chi^2}\sum_{n=0}^\infty \chi^n|n\rangle_{A}\otimes|n\rangle_{A'}, \ \chi=\tanh{r},
\end{align}
where $r$ is the squeezing parameter. Let one mode of the TMSV state be transmitted through a pure loss channel of transmissivity $\eta$, followed by $N-$quantum scissors, as shown in Fig.~\ref{half channel}, where $N$ is sufficiently large, so that it implements NLA of gain $g=\sqrt{(1-\kappa)/\kappa}$. The NLA results in a heralded state that is equivalent to the state obtained by transmitting one mode of a TMSV state of higher mean photon number $\mu'=\sinh^2 (\tanh^{-1}(\chi\sqrt{1+(g^2-1)\eta}))$ through a pure loss channel of improved transmissivity given by $\eta'=g^2\eta/\left(1+(g^2-1)\eta\right),$ as shown in~\cite{RL09}, which implies improved entanglement shared across the channel.

\section{CV entanglement distillation with multiple quantum scissors: Methods \label{sec:Entanglement-distillation-over}}

Our analysis of CV entanglement distillation across a pure loss bosonic channel using multiple, but finite $N-$quantum scissors with TMSV input as shown in Fig.~\ref{half channel} involves the following steps: a) determining the heralded non-Gaussian quantum state, b) determining the heralding success probability, and c) evaluating a figure of merit for the task.
\begin{figure}[H]
	\medskip{}
	\centering
	
	\begin{tabular}{c}
		\includegraphics[scale=0.575]{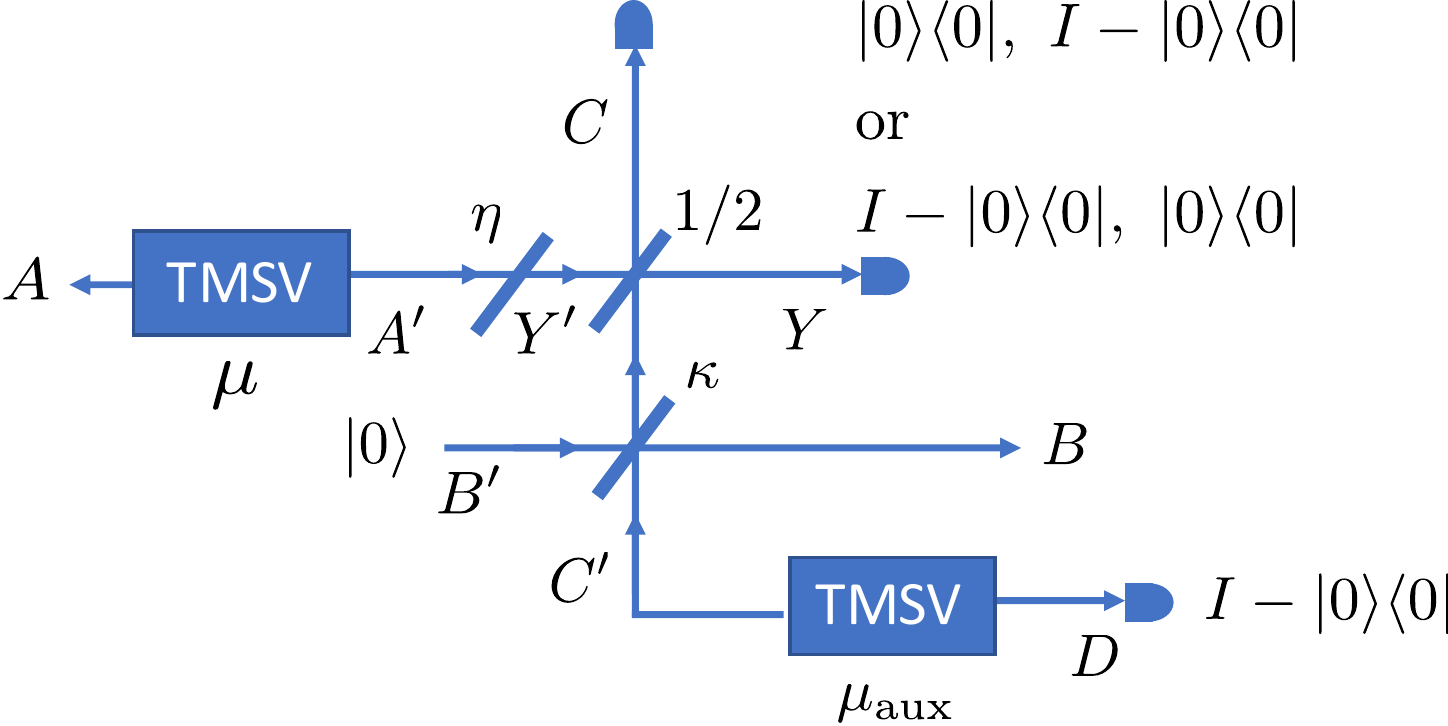}\tabularnewline
	\end{tabular}
	
	\caption{A modified quantum scissors implemented
		with a weak EPR source and ON-OFF projections, appended to a pure loss channel of transmissivity $\eta$ with TMSV state input. When the idler mode of the weak TMSV state is projected on $I-\left|0\right\rangle \left\langle 0\right|$ a single photon is heralded in the signal mode and injected into the
		quantum scissors.}
	\label{HC with mod scissors}
\end{figure}

\subsection{Quantum scissors based on heralded single photon injection
and ON-OFF photodetection}

The calculation of the non-Gaussian state heralded by $N-$quantum scissors of the form shown in Fig.~\ref{LOPD NLA} (b) quickly becomes cumbersome with increasing (but finite) $N$ (c.f.~\cite{DR2018,DR2017}). In order
to make the calculation more tractable, we emulate the quantum scissors
by replacing the photon number detectors with ON-OFF photodetection
$\left\{ \left|0\right\rangle \left\langle 0\right|,\,I-\left|0\right\rangle \left\langle 0\right|\right\} $,
and the single photon by one mode of a weak TMSV state (of mean photon
number $\mu_{\textrm{aux}}$). Figure~\ref{HC with mod scissors} shows
such an emulation of a single-quantum scissors NLA (denoted as $N=1-$NLA
hereafter) acting on a pure loss channel with TMSV input. The weak
TMSV state heralds a single photon state for injection into the quantum
scissors when the idler mode is projected onto $I-\left|0\right\rangle \left\langle 0\right|$.
Note that in all the calculations presented in
this paper, we choose $\mu_{\textrm{aux}}=0.01$. The success probability
of the quantum scissors operation is expected to be enhanced having
replaced the $\left\{ \left|0\right\rangle ,\left|1\right\rangle \right\} $
projection with $\left\{ \left|0\right\rangle \left\langle 0\right|,\,I-\left|0\right\rangle \left\langle 0\right|\right\} $
detection, whereas the quality of the amplification is expected to
be marginally degraded. However, the qualitative behavior of the scheme
still remains preserved as is shown in subsequent analyses.

\subsection{Heralded state and heralding success probability\label{subsec:Heralded-state-calc}}

We now describe the calculation of the heralded non-Gaussian quantum
state and heralding success probability for the scheme shown in Fig.~\ref{HC with mod scissors}. The scheme involves a bosonic system of five
modes, whose pre-measurement state is Gaussian, meaning the quantum
state is completely described by its first two moments. See Appendix~\ref{sec:Gaussian-States,-Unitaries} for the system description, and the initial and pre-measurement covariance matrices. See~\cite{Serafini17, WPGCRSL12} for a detailed account of CV quantum information including entanglement in CV Gaussian states. 

The quantum scissors operation is successful when, in Fig.~\ref{HC with mod scissors},
either one (but not both) of the modes $C$ and $Y,$ along with mode
$D$ are measured in the ``ON'' projection $I-\left|0\right\rangle \left\langle 0\right|.$
The heralded state clearly is non-Gaussian since the $I-\left|0\right\rangle \left\langle 0\right|$
projection is non-Gaussian. We capture the heralded non-Gaussian state
in modes $A$ and $B$ by its Husimi $Q$ function, defined as 
\begin{equation}
Q_{\rho}\left(\alpha,\beta\right)=\frac{\left\langle \alpha,\beta\right|\rho\left|\alpha,\beta\right\rangle }{\pi^{2}},\:\alpha,\beta\in\mathbb{C}.\label{eq: Q fn def}
\end{equation}
It can be determined as an overlap integral between the five mode
Gaussian state $\rho$ in modes $ABCDY$ and the projections $\left|\alpha\right\rangle \left\langle \alpha\right|_{A}\otimes\left|\beta\right\rangle \left\langle \beta\right|_{B}\otimes\left(I-\left|0\right\rangle \left\langle 0\right|\right)_{CD}^{\otimes2}\otimes\left|0\right\rangle \left\langle 0\right|_{Y}$
or $\left|\alpha\right\rangle \left\langle \alpha\right|_{A}\otimes\left|\beta\right\rangle \left\langle \beta\right|_{B}\otimes\left|0\right\rangle \left\langle 0\right|_{C}\otimes\left(I-\left|0\right\rangle \left\langle 0\right|\right)_{DY}^{\otimes2},$
normalized by the probability of the projections $\pi_{1}=\left(I-\left|0\right\rangle \left\langle 0\right|\right)_{CD}^{\otimes2}\otimes\left|0\right\rangle \left\langle 0\right|_{Y}$
and $\pi_{2}=\left|0\right\rangle \left\langle 0\right|_{C}\otimes\left(I-\left|0\right\rangle \left\langle 0\right|\right)_{DY}^{\otimes2},$
respectively, which constitute the success probability of the quantum
scissors operation. The heralded states corresponding to the two possible
successful projections turn out to be the same up to local phases.

The above overlap integrals are sums of Gaussian integrals that can
be evaluated efficiently (Appendix~\ref{sec:Non-Gaussian-Measurement}).
For example, the success probability for the projection $\pi_{1}$
on modes $CDY$ involves evaluating the overlap integral $P_{1}=\operatorname{Tr}\left(\pi_{1}\rho_{ABCDY}\right),$
\begin{align}
P_{1} & =\int d\xi\chi_{\rho_{ABCDY}}\left(\xi\right)\left(1-\chi_{0}\left(-\xi_{c}\right)\right)\nonumber \\
 & \times\left(1-\chi_{0}\left(-\xi_{D}\right)\right)\left(\chi_{0}\left(-\xi_{Y}\right)\right)_{Y},\ \xi\in\mathbb{R},\label{eq: NLA heralding prob v1}
\end{align}
where $\chi_{\rho_{ABCDY}}$ and $\chi_{0}=\chi_{\left|0\right\rangle \left\langle 0\right|}$
are the characteristic functions of the heralded non-Gaussian state
and the vacuum state, respectively. The success probability associated
with the other projection, namely $\pi_{2}$ also turn out to be the
same as (\ref{eq: NLA heralding prob v1}) due to symmetry between
modes $C$ and $Y,$ so that the total success probability is $P_{\textrm{succ}}'=\operatorname{Tr}\left(\left(\pi_{1}+\pi_{2}\right)\rho\right)=P_{1}+P_{2}=2P_{1}$.
The success probability of a quantum scissors operation with a deterministic
single photon injection into mode $C$ can be deduced from the above
success probability by renormalizing it with the probability of detecting
a single photon in the idler mode $D.$ That is, 
\begin{equation}
P_{\textrm{succ}}=P'_{\textrm{succ}}\bigg/\left(\frac{\mu_{\textrm{aux}}}{\mu_{\textrm{aux}}+1}\right),\label{eq: NLA heralding prob v2}
\end{equation}
where the scaling factor is the probability of observing the $\left|1\right\rangle \left\langle 1\right|$
projection on the idler mode of the TMSV state of mean photon number
$\mu_{\textrm{aux}}.$

The heralded state and heralding success probability calculations
for NLA with higher number of quantum scissors follow similarly to
the $N=1-$NLA case described above. For $N$ quantum scissors, the
success probability $P_{\textrm{succ}}$ is obtained by renormalizing
$P'_{\textrm{succ}}$ by a factor $\left(\nicefrac{\mu_{\textrm{aux}}}{\left(\mu_{\textrm{aux}}+1\right)}\right)^{N}.$

\subsection{Reverse coherent information}

The reverse coherent information (RCI) of a state $\rho_{AB}$ is defined as \cite{GPLS09,PGBL09,DJKR06,Hcube00}
\begin{equation}
I_{R}\left(\rho_{AB}\right):=H\left(A\right)_{\rho}-H\left(AB\right)_{\rho},\label{eq: RCI definition}
\end{equation}
where $H(A)_{\rho}$ is the von Neumann entropy of $\rho_A=\operatorname{Tr}_B(\rho_{AB})$ defined as $H\left(A\right)_{\rho}=-\operatorname{Tr}(\rho_A\log_2\rho_A)$ (and likewise $H(AB)_{\rho_{AB}}$).

For the CV entanglement distillation scheme of Fig.~\ref{HC with mod scissors}, we determine the optimal RCI that can be heralded by numerically optimizing over the mean photon number of the input TMSV state and the NLA gain~\footnote{See Supplemental Material for the code repository.}.

\subsection{A lower bound on the heralded reverse coherent information}

Evaluating the RCI of the state heralded upon successful operation
of quantum scissors-based NLA following the pure loss channel is perceived
to be nontrivial. As an interim remedy, we resort to calculating the
RCI of the covariance matrix of the heralded state, which by the Gaussian
extremality theorem \cite{WGC06} amounts to a lower bound on the
RCI of the heralded non-Gaussian state. 

\begin{figure}
	\begin{tabular}{c}
		\includegraphics[scale=0.24]{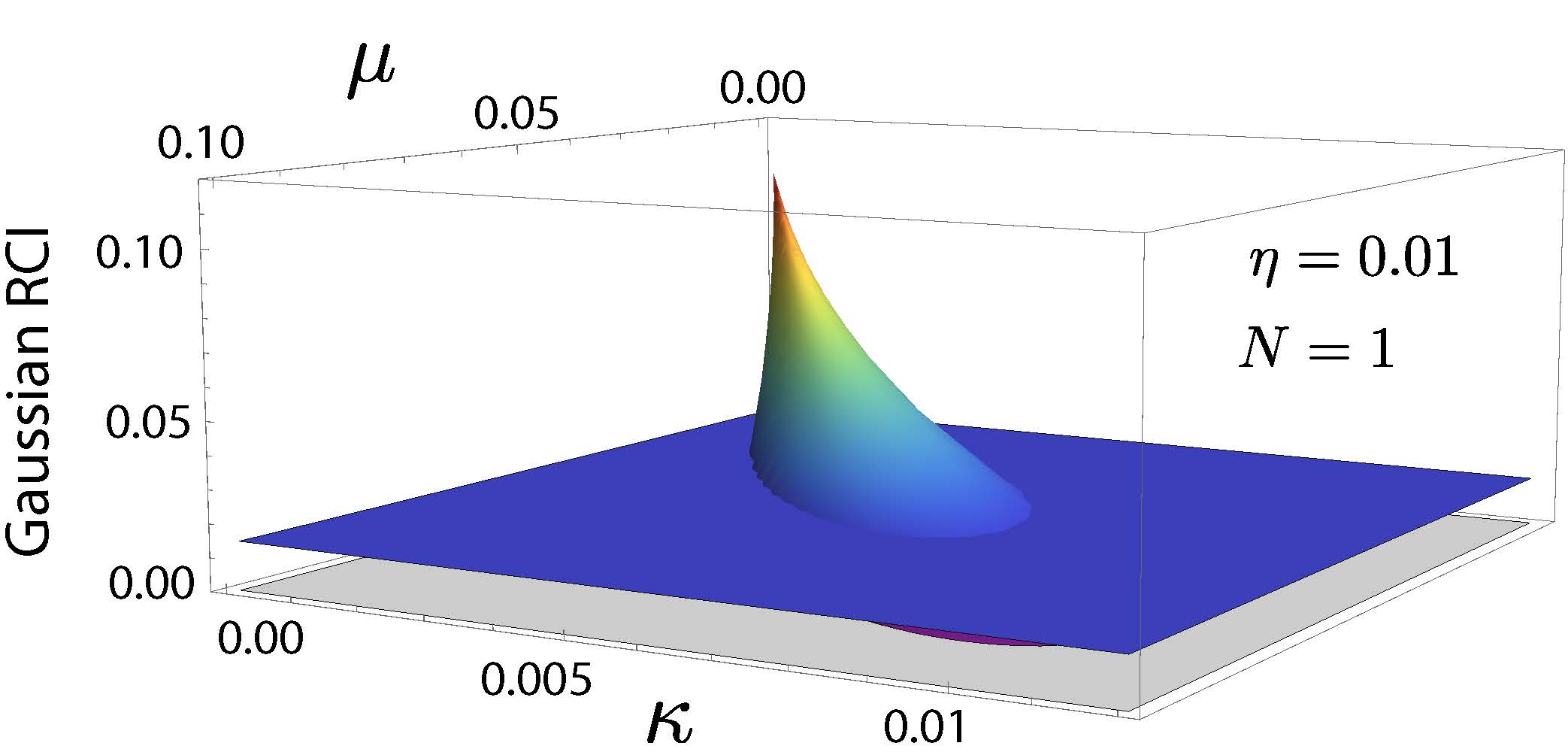}\tabularnewline
	\end{tabular}
	
	(a)
	
	\begin{tabular}{c}
		\includegraphics[scale=0.24]{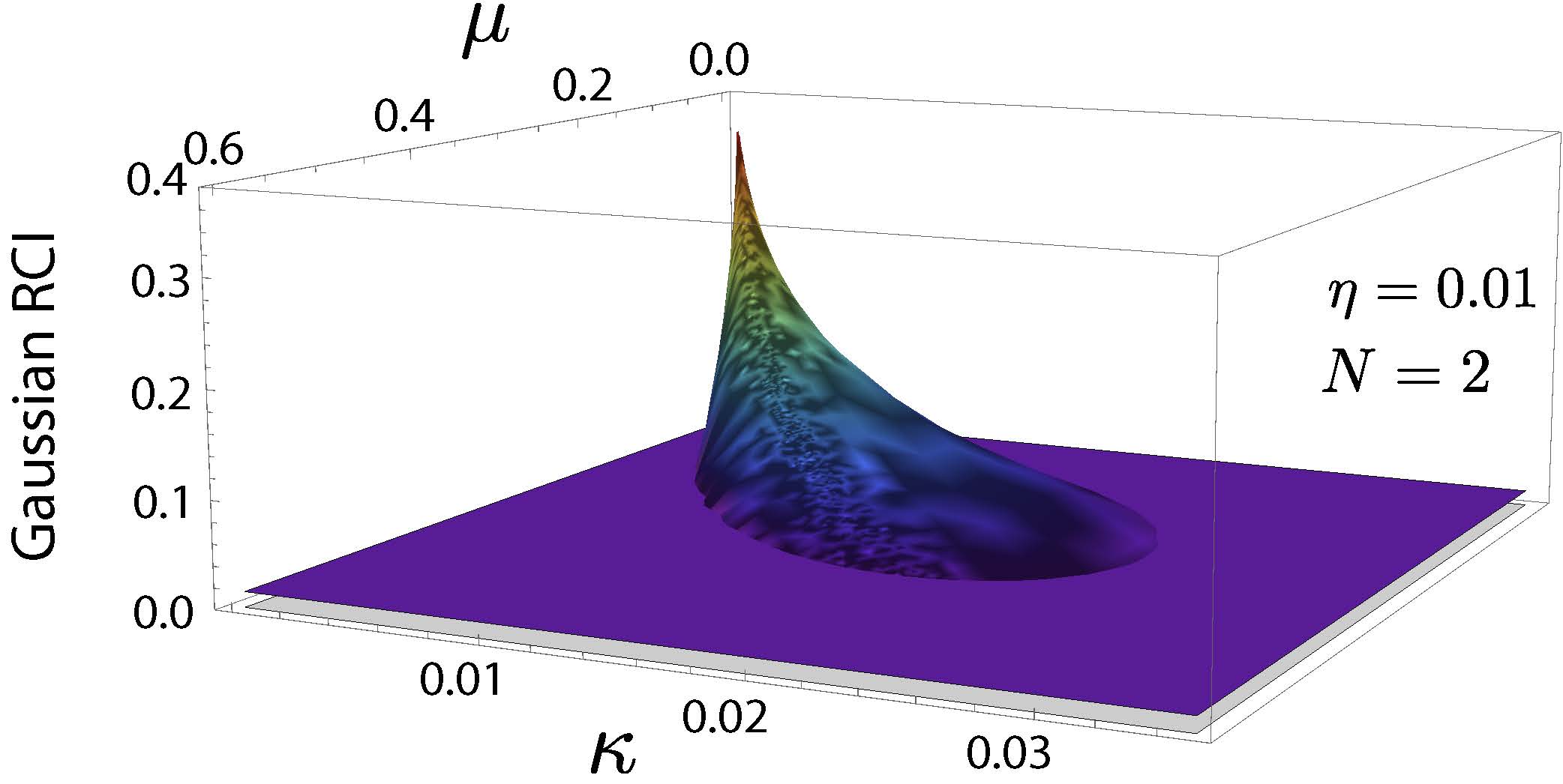}\tabularnewline
	\end{tabular}
	
	(b)
	
	\centering\caption{Gaussian RCI heralded across a pure loss channel of transmissivity
		$\eta=0.01$ appended with (a) $N=1-$NLA, (b) $N=2-$NLA as a function of the input
		TMSV state mean photon number and the gain of the NLA. The parameter
		$\kappa$ is related to the NLA gain by $\kappa=1/(1+g^{2}).$ The
		floors of the plots correspond to $C_{\rm direct}(\eta).$ The roughness in the surface in (b) is due to the finiteness of precision in our numerics.} 
	\label{rcivsmukappa01}
\end{figure}

\begin{figure}
	\begin{tabular}{c}
		\includegraphics[scale=0.22]{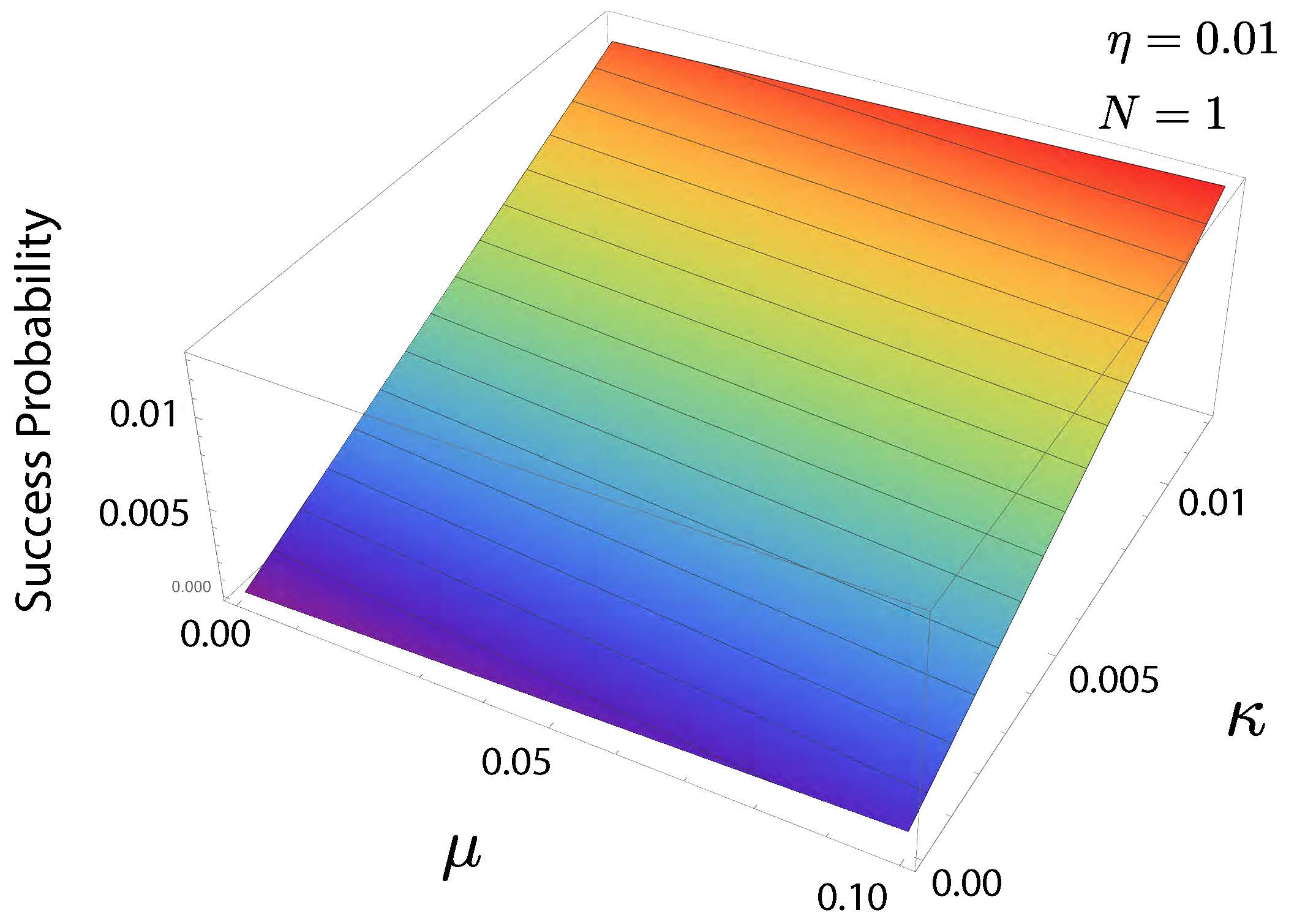}\tabularnewline
	\end{tabular}
	
	(a)
	
	\begin{tabular}{c}
		\includegraphics[scale=0.22]{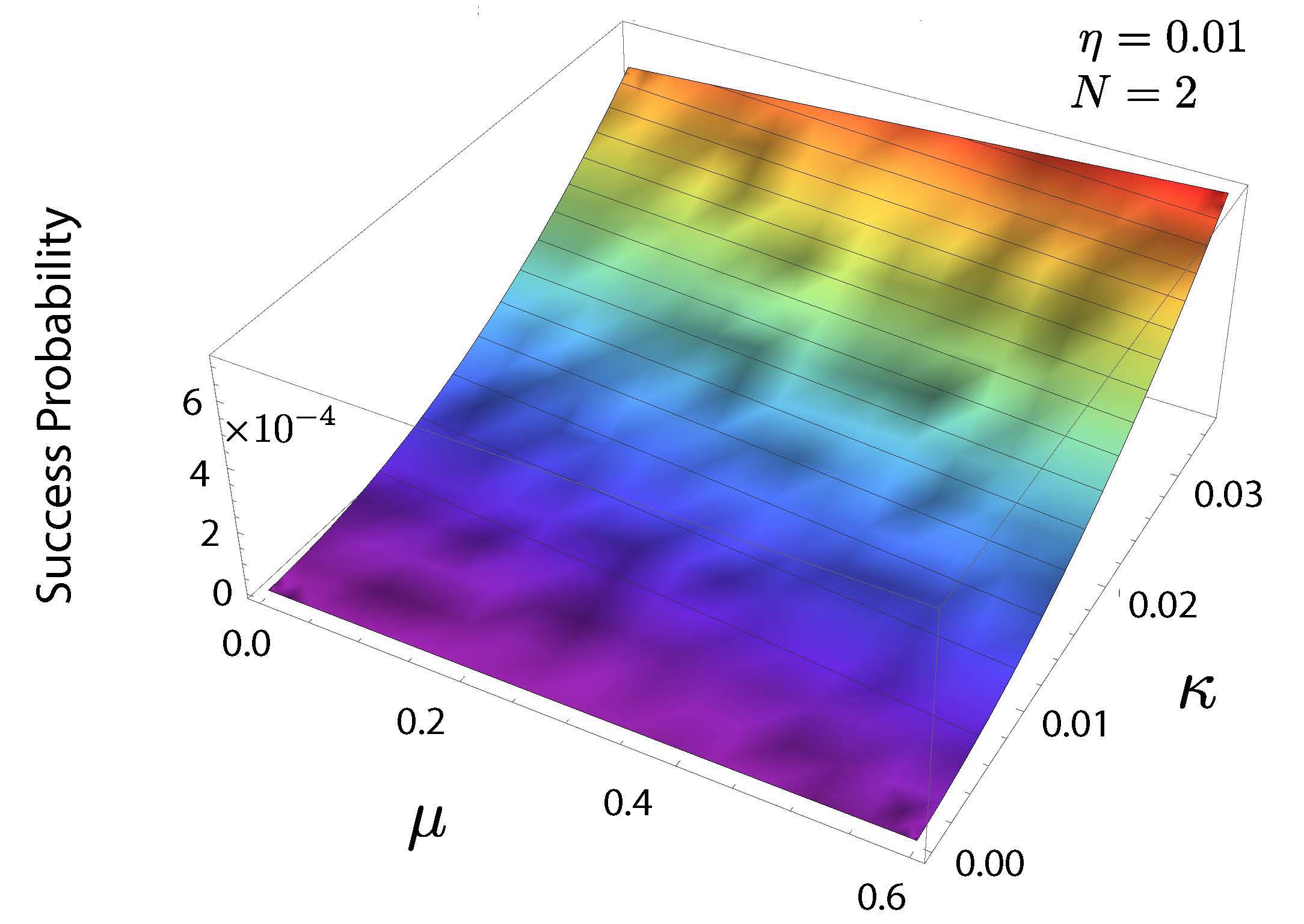}\tabularnewline
	\end{tabular}
	
	(b)
	
	\centering\caption{Heralding success probability with $N=1,2-$ NLA deployed on a pure
		loss channel of transmissivity $\eta=0.01$, as a function of the
		input TMSV state mean photon number and the gain of the NLA. The parameter
		$\kappa$ is related to the NLA gain by $\kappa=1/(1+g^{2}).$}
	\label{psuccvsmukappa01}
\end{figure}

The covariance matrix $V\left(\rho\right)$ can be determined from
the $Q$ function as
\begin{align}
V_{i,j}\left(\rho\right) & =2\intop d\mathbf{r}r_{i}r_{j}Q_{\rho}\left(\alpha,\beta\right)-\delta_{i,j}\nonumber \\
 & -2\intop d\mathbf{r}r_{i}Q_{\rho}\left(\alpha,\beta\right)\intop d\mathbf{r}r_{j}Q_{\rho}\left(\alpha,\beta\right),\label{eq: heralded cov from Q}
\end{align}
where $\mathbf{r}=\left(x_{1},x_{2},p_{1},p_{2}\right)^{\operatorname{T}}\in\mathbb{R}^{4},$
$\alpha=\left(x_{1}+ip_{1}\right)/\sqrt{2},\:\beta=\left(x_{2}+ip_{2}\right)/\sqrt{2},$
and the $Q$ function in real coordinates is $Q_{\rho}\left(x_{1},x_{2},p_{1},p_{2}\right)=\left\langle x_{1},x_{2},p_{1},p_{2}\right|\rho\left|x_{1},x_{2},p_{1},p_{2}\right\rangle /\left(4\pi^{2}\right).$

Given the covariance matrix $V\left(\rho_{AB}\right)$ of a bipartite
state $\rho_{AB},$ the RCI of $V\left(\rho_{AB}\right)$ follows
from (\ref{eq: RCI definition}) as 
\begin{align}
I_{R}\left(A'\left\langle B\right.\right)_{\rho_{AB}} & =H\left(A\right)_{\rho_{AB}}-H\left(AB\right)_{\rho_{AB}}\nonumber \\
 & =g\left(\vec{\nu}_{A}\right)-g\left(\vec{\nu}_{AB}\right),\label{eq: Gauss RCI}
\end{align}
where $g\left(x\right):=\left(\frac{x+1}{2}\right)\log_{2}\left(\frac{x+1}{2}\right)-\left(\frac{x-1}{2}\right)\log_{2}\left(\frac{x-1}{2}\right)$
is the entropy of a thermal state of mean photon number $\left(x-1\right)/2,$
and $\vec{\nu}_{A}$ and $\vec{\nu}_{AB}$ are the symplectic eigenvalues
of the covariance matrices corresponding to mode $A$ and modes $AB,$
respectively (Appendix~\ref{sec:Gaussian-States,-Unitaries}). We call this the Gaussian RCI of the state $\rho_{AB}.$

\section{Results\label{sec:Results}}

\begin{figure}
\begin{tabular}{c}
 \includegraphics[scale=0.28]{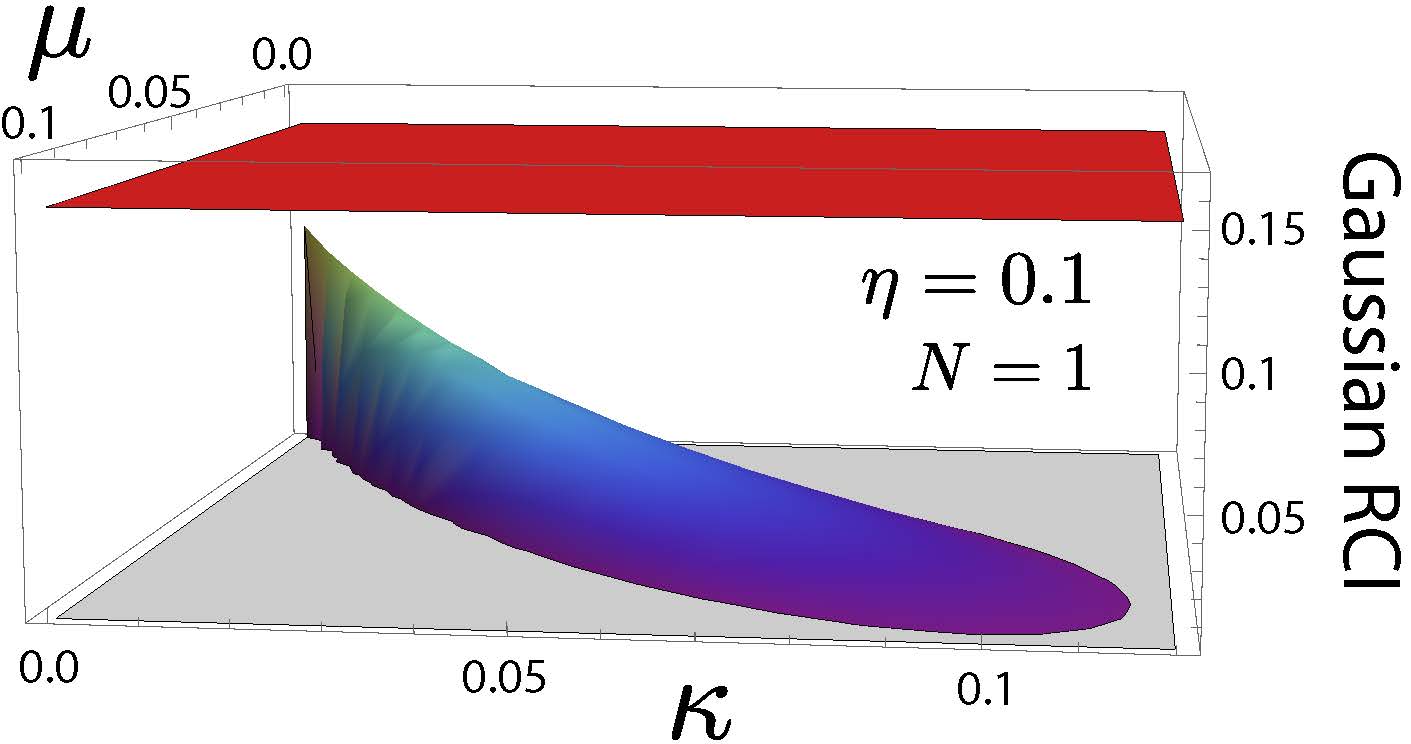}\tabularnewline
\end{tabular}

(a)

\begin{tabular}{c}
 \includegraphics[scale=0.28]{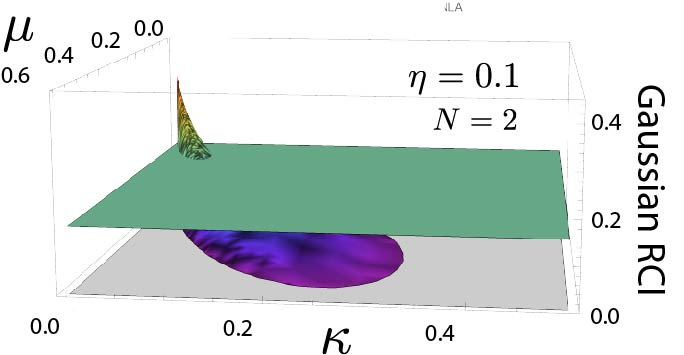}\tabularnewline
\end{tabular}

(b)
\centering\caption{Gaussian RCI heralded across a pure loss channel of transmissivity
	$\eta=0.1$ appended with $N=1,2-$NLA, as a function of the input
	TMSV state mean photon number and the gain of the NLA. The parameter
	$\kappa$ is related to the gain of the NLA by $\kappa=1/(1+g^{2}).$
	The planes in the plots correspond to $C_{\rm direct}(\eta).$ }
\label{rcivsmukappa_activation}

\end{figure}

\begin{figure}
\begin{tabular}{c}
 \includegraphics[scale=0.48]{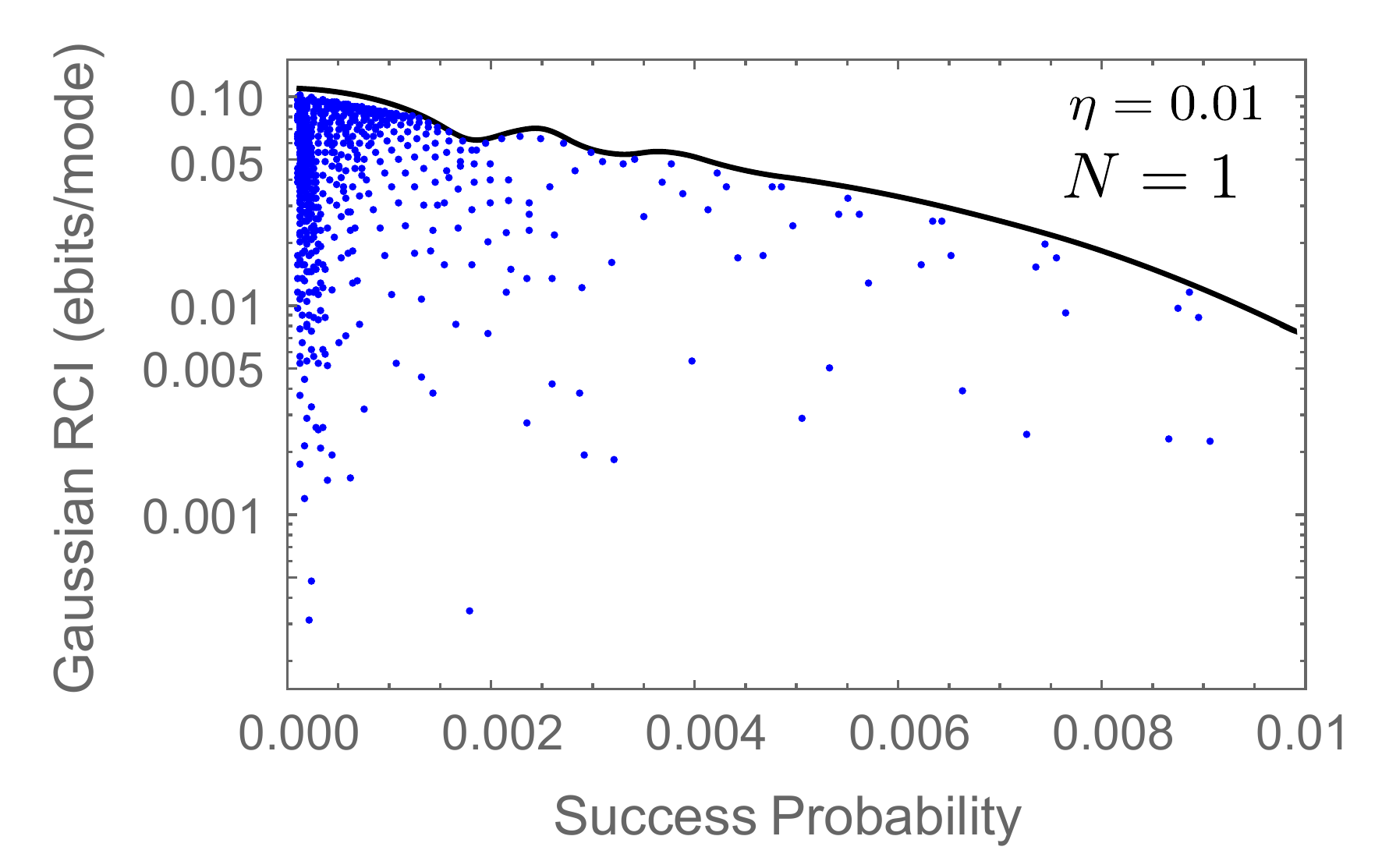}\tabularnewline
\end{tabular}

(a)

\begin{tabular}{c}
 \includegraphics[scale=0.48]{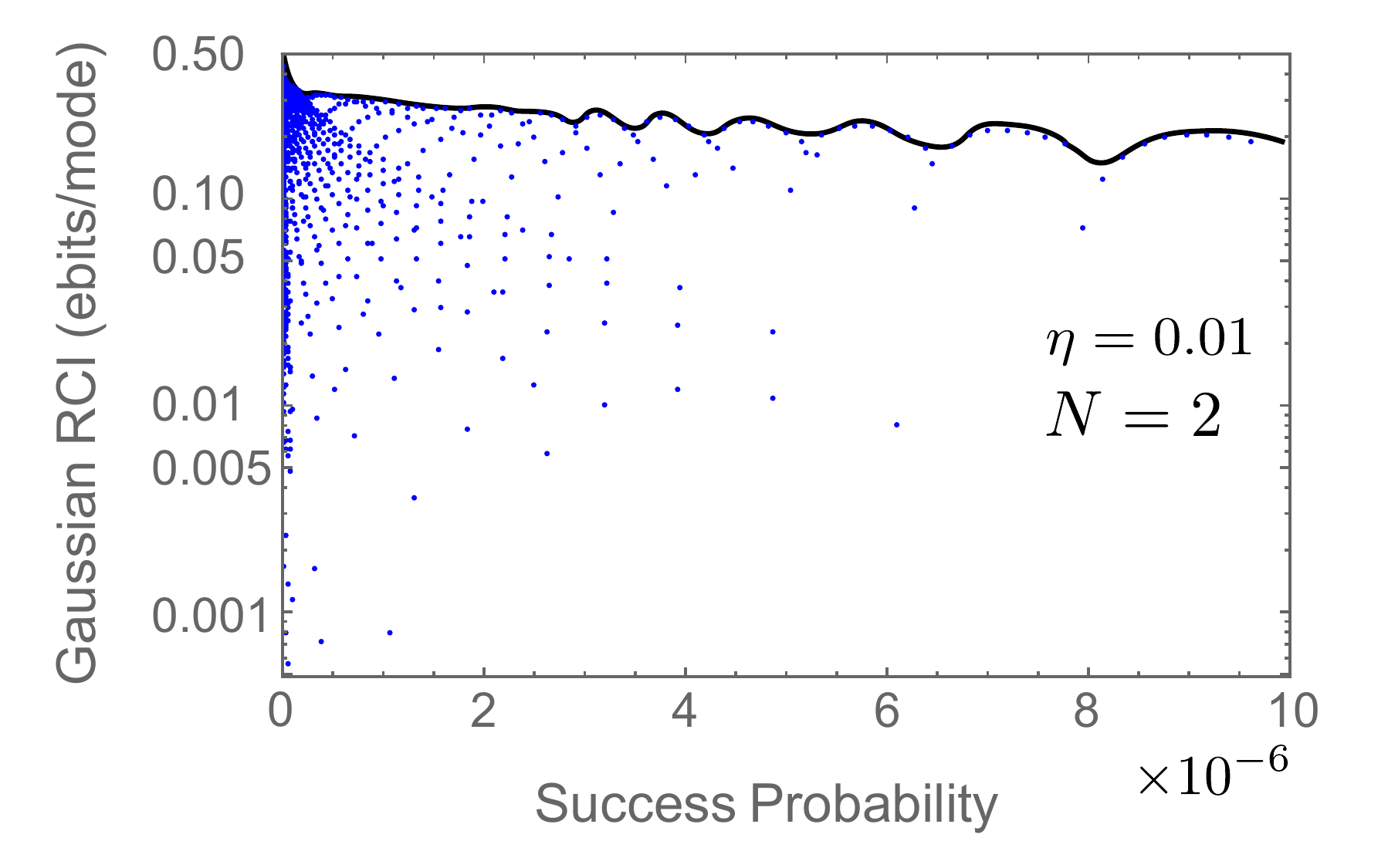}\tabularnewline
\end{tabular}

(b)

\centering\caption{Gaussian RCI heralded across $N=1,2-$NLA vs heralding success probability
	for a pure loss channel of transmissivity $\eta=0.01$.}
\label{HCgaussRCIscatter}

\end{figure}

\subsection{Gaussian reverse coherent information for $N=1,2-$NLA-assisted pure
loss channel}

In Figs.~\ref{rcivsmukappa01} and \ref{psuccvsmukappa01}, we plot
the Gaussian RCI and the heralding success probability for a pure loss channel appended with $N=1-$NLA and $N=2-$NLA.
The quantity is plotted as a function of the mean photon number $\mu$
of the TMSV input and the transmissivity $\kappa=1/(1+g^{2})$ of
the asymmetric beamsplitter in the quantum scissors ($g$ being the
NLA gain). We choose a channel of transmissivity $\eta=0.01.$ The
choice of a small transmissivity $\eta$ for illustration is such that
the amplification due to quantum scissors remains noiseless. The mean
photon number of the input TMSV state is optimized to determine the
best possible heralded Gaussian RCI. We make the following observations
from these Figures. a) The heralded Gaussian RCI exceeds the direct
transmission capacity $C_{\rm direct}(\eta)$ of Eq.~(\ref{plob}) (denoted by
the floor of both the 3-d plots in Fig.~\ref{rcivsmukappa01}) for
a certain regime of the parameters $\mu$ and $\kappa$ for both $N=1,2.$ b) The Gaussian RCI of $N=2-$NLA is about four times that of $N=1-$NLA. c) The increase
in the Gaussian RCI is accompanied by a steep decrease in the heralding
success probability of the NLA, which drops exponentially going from
$N=1$ to $N=2-$NLA as shown in Fig.~\ref{psuccvsmukappa01}.

For a different choice of $\eta,$ namely $\eta=0.1$, the optimal
heralded Gaussian RCI with $N=1-$NLA never exceeds $C_{\rm direct}(\eta),$
whereas with $N=2-$NLA it exceeds the bound, as shown in Fig.~\ref{rcivsmukappa_activation}. In other words, in such a parameter regime one quantum scissors is not enough, and a second quantum scissors is required to ``activate'' a heralded entanglement of higher quality than direct transmission. 

Since the increase in the heralded Gaussian RCI happens at the expense
of the decreased heralding success probability, it is important to
characterize this tradeoff so that the values of the parameters (TMSV
input mean photon number $\mu$ and NLA gain $g$) can be chosen optimally.
In Fig.~\ref{HCgaussRCIscatter}, we plot the Gaussian RCI as a function
of the success probability for the $N=1,2-$NLA-appended channel of
transmissivity $\eta=0.01,$ when $\mu$ and $g$ are optimized. The
outer envelope of this scatter plot thus represents the best possible
pairs of heralded RCI and heralding success probability attainable
using $N=1,2-$NLA. The oscillations in the outer envelope are a result of focusing on the Gaussian part (covariance matrix) of the heralded state, whose Q function is a sum of many Gaussian distributions and thus non-Gaussian. The envelope and the accompanying set of optimal
parameters thus may be of use in designing quantum repeater schemes
with quantum scissors based NLA. 

\subsection{Entanglement of formation and reverse coherent information lower
bounds for teleportation through NLA-assisted pure loss channel}

In \cite{Ralph2011}, a variant of the scheme in Fig. \ref{half channel}
was studied as quantum error correction for the transmission of quantum
continuous variable states over a lossy channel. The scheme is as
shown in Fig.~\ref{Ralph EC box with mod scissors}. Here, direct
transmission through the lossy channel is replaced by continuous-variable
quantum teleportation over a lossy entangled resource established
by sending one mode of a finite energy two-mode squeezed vacuum (TMSV)
state through the channel followed by NLA. The resource is referred
to as an error correction (EC) box, and is characterized by the mean
photon number of the teleportation resource TMSV state $\mu_{\textrm{res}},$
the number of quantum scissors $N,$ and the NLA gain $g.$ Using
a meticulous Fock basis calculation, Dias and Ralph \cite{DR2018}
recently showed that this scheme with $N=1-$NLA, when NLA is successful,
can herald states with higher entanglement than the state shared via direct transmission over the lossy channel. They considered the entanglement of formation (EOF)
(Appendix~\ref{sec:Entanglement-of-Formation}) of the heralded
covariance matrix (Gaussian part of the heralded non-Gaussian state)
as the figure of merit, which constitutes a lower bound on the entanglement
of formation of the heralded non-Gaussian state by the Gaussian extremality
theorem~\cite{WGC06} . 

\begin{figure}
\medskip{}

\begin{tabular}{c}
 \includegraphics[scale=0.5]{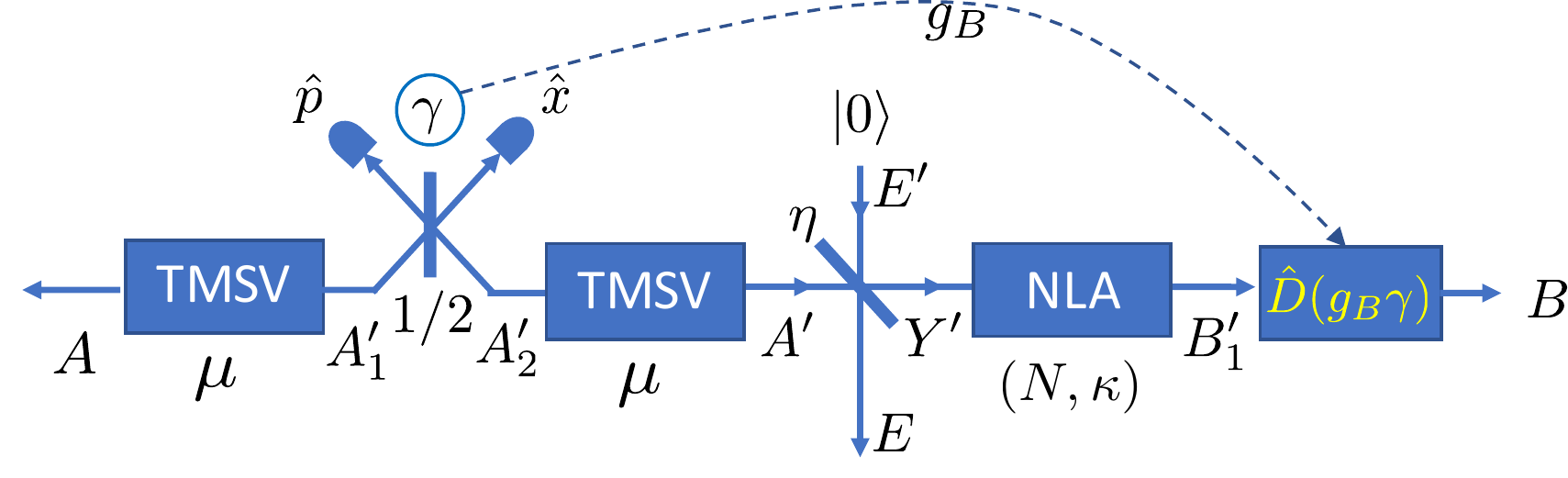}\tabularnewline
\end{tabular}

\centering\caption{One mode of a TMSV state teleported through a resource state consisting
of a lossy TMSV state aided by a $N=1-$NLA.}
\label{Ralph EC box with mod scissors}
\end{figure}

\begin{figure}
\begin{tabular}{c}
 \includegraphics[scale=0.5]{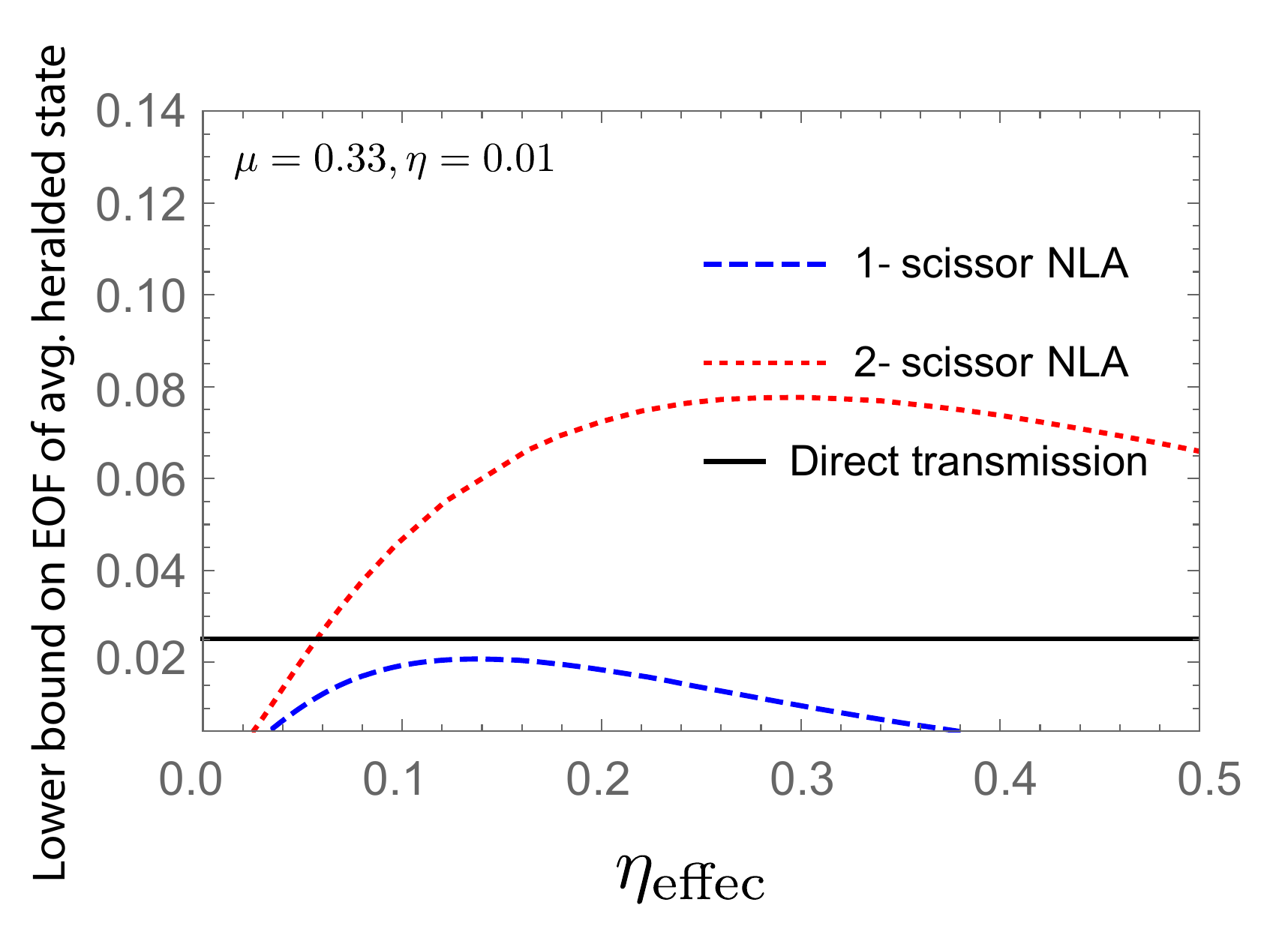}\tabularnewline
\end{tabular}

(a)

\begin{tabular}{c}
 \includegraphics[scale=0.5]{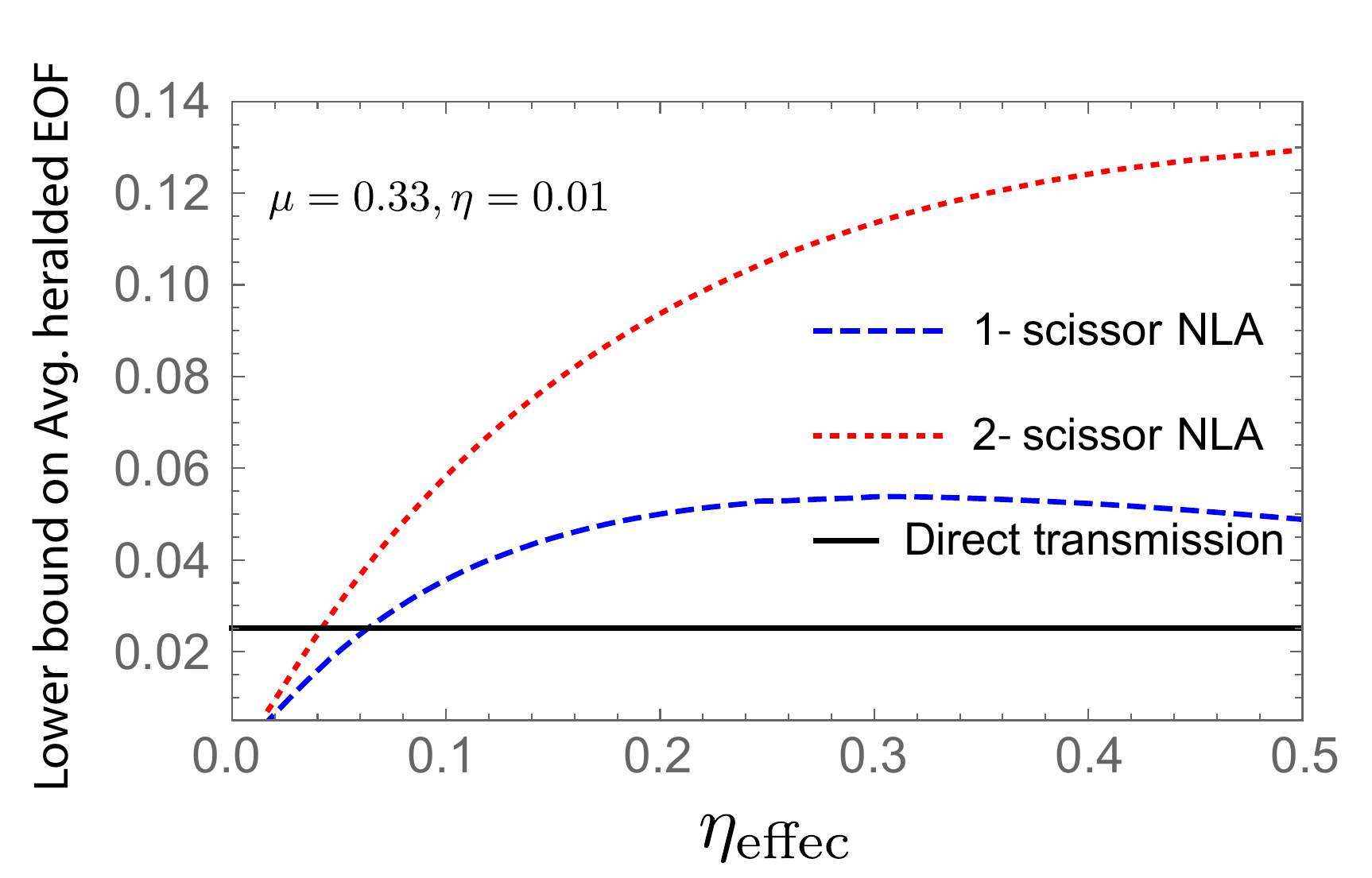}\tabularnewline
\end{tabular}

(b)

\centering\caption{Entanglement of formation (EOF) of the covariance matrix heralded
across the EC Box of \cite{Ralph2011} for a TMSV input, as a function
of the heralded effective transmission. The bare channel transmissivity
is chosen to be $\eta=0.01,$ mean photon number of the input TMSV
and the teleportation resource TMSV chosen to be equal $\mu_{\textrm{res}}=\mu=0.33,$
$N=1,2$, and the NLA gain $g$ is varied. a) The EOF of the covariance
matrix of the average state heralded across the error corrected channel.
b) The average of EOF of the conditional heralded covariance matrices
(conditioned on and averaged over the outcome of the dual homodyne
detection). The black bold line corresponds to the EOF for transmission across
the bare lossy channel. }
\label{Ralphgeof}
\end{figure}

\begin{figure}
\begin{tabular}{c}
 \includegraphics[scale=0.45]{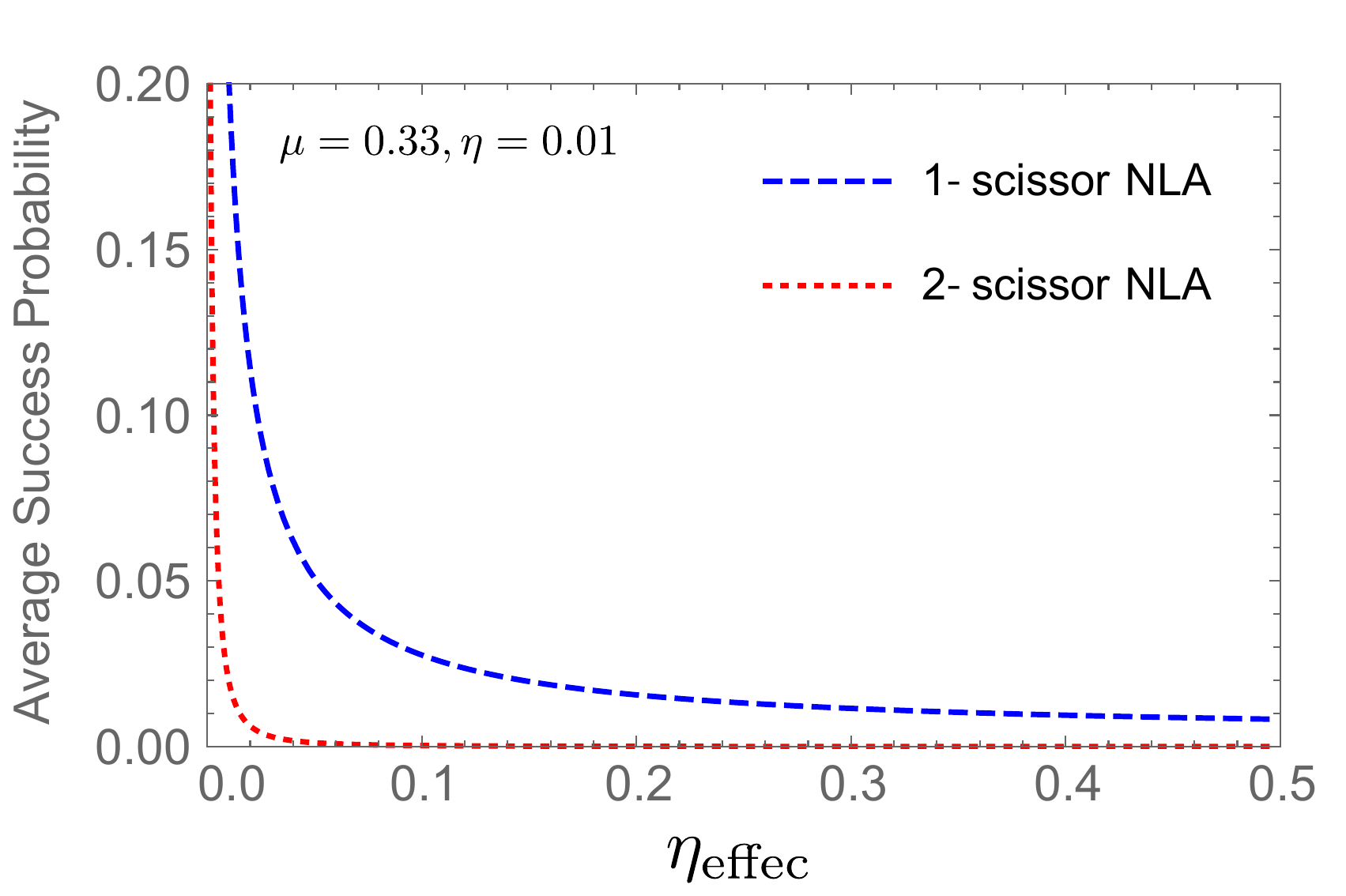}\tabularnewline
\end{tabular}

\centering\caption{Success Probability of \cite{Ralph2011}'s EC box for $N=1,2-$NLA
as a function of the heralded effective transmission. The bare channel
transmissivity is chosen to be $\eta=0.01,$ mean photon number of
the input TMSV and the teleportation resource TMSV chosen to be equal
$\mu_{\textrm{res}}=\mu=0.33,$ $N=1,2$, and the NLA gain $g$ is
varied.}
\label{Ralphsuccessprob}
\end{figure}

\begin{figure}
\begin{tabular}{c}
 \includegraphics[scale=0.45]{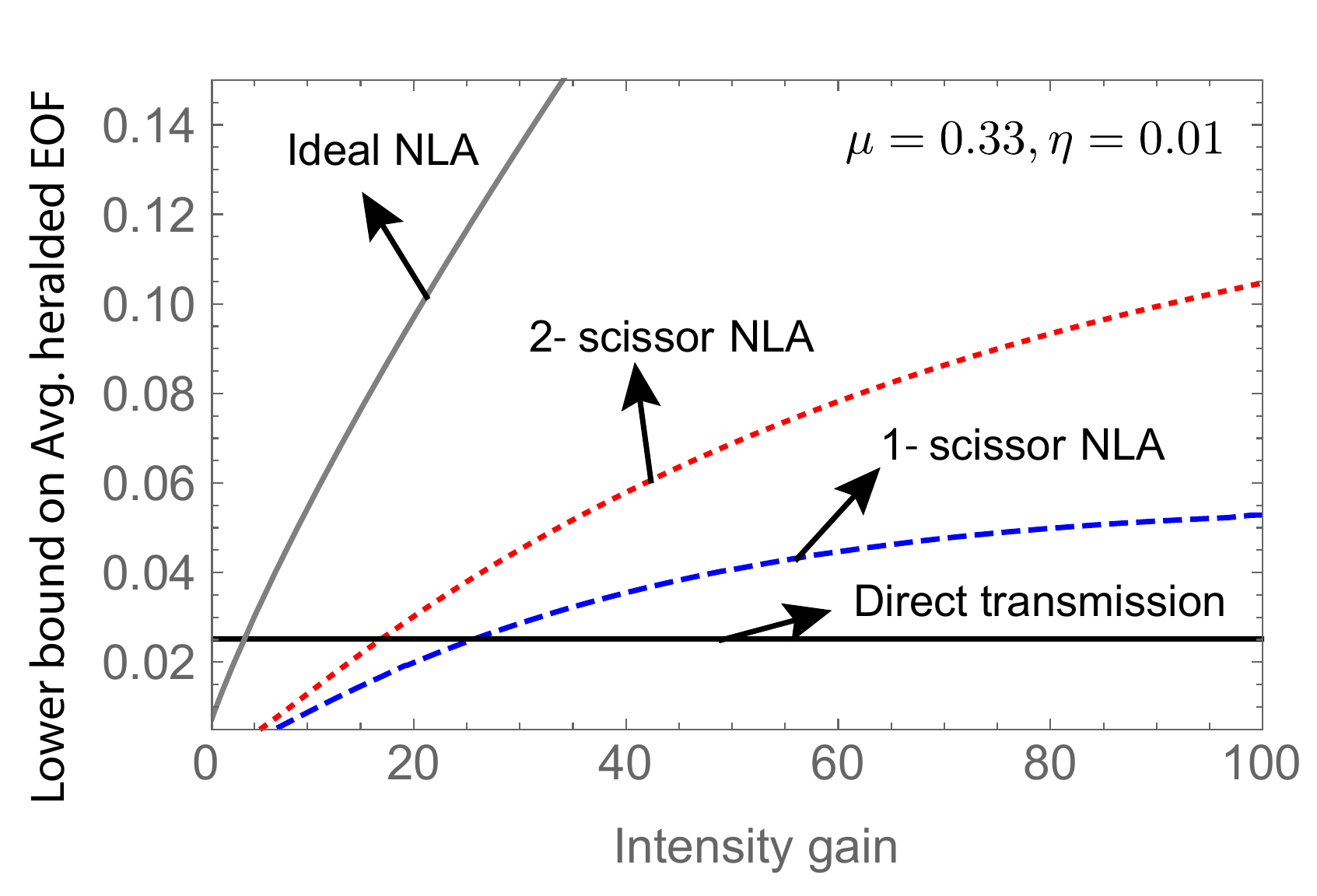}\tabularnewline
\end{tabular}

(a)

\begin{tabular}{c}
 \includegraphics[scale=0.45]{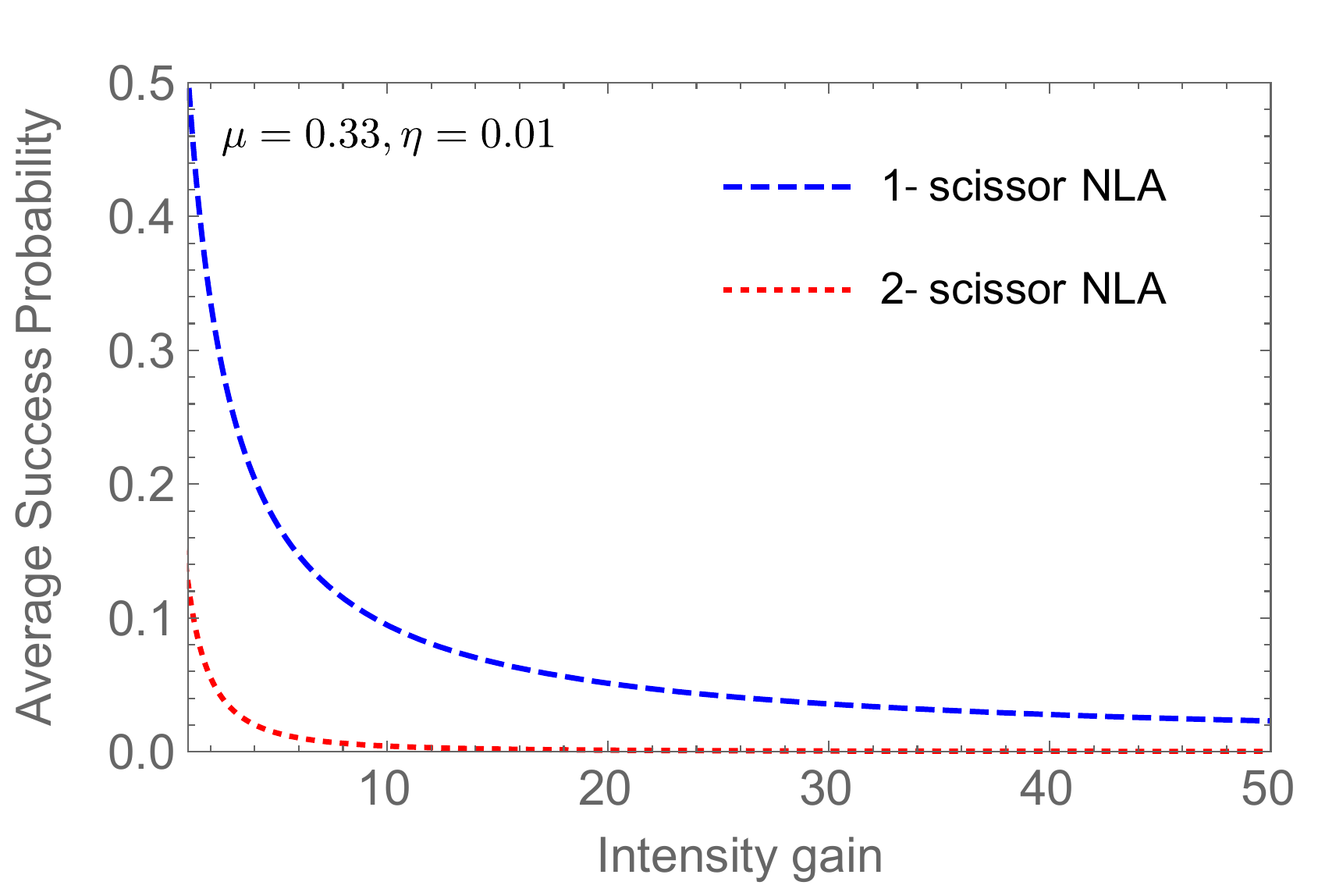}\tabularnewline
\end{tabular}

(b)

\centering\caption{EOF and success probability of the covariance matrix heralded by EC box for $N=1,2-$NLA as a function
of the NLA intensity gain $g^{2}$. }
\label{Ralphgeofpsuccvsgsq}
\end{figure}

\begin{figure}
\begin{tabular}{c}
 \includegraphics[scale=0.60]{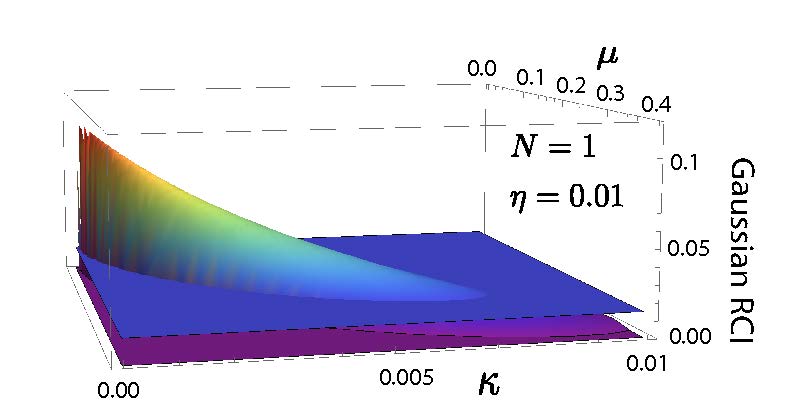}\tabularnewline
\end{tabular}

(a)

\begin{tabular}{c}
 \includegraphics[scale=0.65]{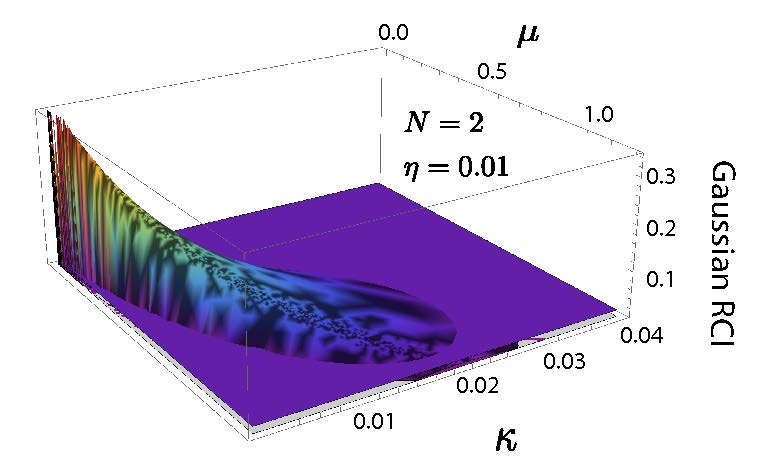}\tabularnewline
\end{tabular}

(b)

\centering\caption{Reverse coherent information of the covariance matrix heralded across
the scheme in Fig.~\ref{Ralph EC box with mod scissors} with $N=1,2-$NLA
and $\eta=0.01$. The parameter $\kappa$ is related to the gain of
the NLA by $\kappa=1/(1+g^{2}).$}
\label{Ralphgrci}
\end{figure}

The calculation of Dias and Ralph, however, is tedious and difficult
to be extended to NLA with multiple quantum scissors. The method based
on characteristic functions described in Sec.~\ref{subsec:Heralded-state-calc}
offers an efficient alternative means to determine the heralded non-Gaussian
state in this case. Using this method along with general ideas from
Gaussian conditional dynamics \cite{GLS16,Serafini17} to deal with
the teleportation elements such as dual homodyne detection (CV Bell
state measurement) and displacement (unitary correction) (Appendix~\ref{sec:Gaussian-Conditional-Dynamics}), we analyzed the effect
of NLA with $N=1,2-$quantum scissors in Fig.~\ref{Ralph EC box with mod scissors}.
We calculated a lower bound on the EOF
of the heralded state by evaluating yet another measure, the Gaussian
entanglement of formation (GEOF), on the Gaussian part of the heralded state,
which equals its EOF. (Appendix~\ref{sec:Entanglement-of-Formation}
for details.)

In this scheme, the non-Gaussian state heralded upon successful NLA
operation for a TMSV input is a function of the dual homodyne detection
outcome. In particular, both its first and second moments are dependent
on the outcome. There are two ways to quantify the performance of
the scheme: a) evaluating the figure of merit on the covariance matrix
of the average heralded state, or, b) evaluating the average of the
figure of merit applied on the conditional heralded covariance matrices,
where in both cases the averaging is with respect to the dual homodyne
measurement outcomes. We will call these $q_{1}$ and $q_{2},$ respectively. Operationally, the former captures the entanglement content of the average state heralded by the scheme, while the latter captures the entanglement content on a per copy basis, averaged over the copies.
Convexity of the EOF implies that $q_{2}\geq q_{1}$. It should be
noted that while the displacement correction associated with the teleportation
impacts the average heralded state and in turn its EOF, it does not
affect the EOF of the conditional heralded states, since the measure
is independent of the first moments.

In Fig.~\ref{Ralphgeof}, we plot $q_{1}$ and $q_{2}$ for a TMSV
input of mean photon number $\mu$ as a function of the effective
transmission parameter defined as $\eta_{\textrm{effec}}=g^{2}\eta\chi^{2},$
where $g$ is related to the gain of the NLA, and $\chi=\tanh\left(\sinh^{-1}\left(\sqrt{\mu_{\textrm{res}}}\right)\right)$.
The channel transmissivity is chosen to be $\eta=0.01,$ and the mean
photon number of the TMSV state at the input as well as in the EC
box (teleportation resource state) are chosen to be $\mu_{\textrm{res}}=\mu=0.33,$
(which corresponds to $\chi=1$), and the NLA amplitude gain $g$
is scanned over. In Fig.~\ref{Ralphgeof} (a), the quantity $q_{1}$ is plotted, where it has been optimized over a classical gain tuning
parameter that scales the dual homodyne outcome prior to the displacement
correction operation. We find that our curve for the $N=1-$NLA is
qualitatively similar, but below what was reported in \cite{DR2018}
for the same. This is as expected, since we have considered ON-OFF
heralding photodetection instead of perfect photon-number-resolving
(PNR) single photon detection in the quantum scissors. In addition,
we now have calculated the same figure of merit also for $N=2-$NLA.
We observe that the Gaussian lower bound on
the EOF of the heralded state increases significantly in going from $N=1$ to $N=2.$ In Fig.~\ref{Ralphgeof} (b),
the quantity $q_{2}$ is plotted, and as is expected due to convexity of the EOF, they
are higher than the corresponding curves in Fig.~\ref{Ralphgeof}
(a). %
In the remainder of this discussion, we will consider the quantity $q_{2}$ as the sole figure of merit for the scheme. 

In Fig.~\ref{Ralphsuccessprob}, the heralding success probability
of the scheme is plotted as a function of the effective transmission
for the same set of parameter values as chosen in Fig.~\ref{Ralphgeof}.
We observe that our curve for the $N=1$ case is qualitatively similar,
but slightly above the one reported in \cite{DR2018}. This is again
consistent with our choice of ON-OFF photodetection in place of single
photon detection in the quantum scissors.

In Fig.~\ref{Ralphgeofpsuccvsgsq}, the quantity $q_2$ and the heralding success probability of NLA are plotted as a
function of the NLA intensity gain $g^{2}.$ The EOF of direct transmission
through the lossy channel forms the benchmark. This is exceeded by
teleportation of the input through EC box with $N=1,2-$NLA. The figure
also shows the performance of ideal NLA (corresponding to $N\rightarrow\infty$) for comparison, which is calculated as the EOF of a TMSV state whose one mode undergoes a pure loss channel of effective transmissivity $\eta_\textrm{effec}=g^2\eta\chi^2$, $g$ being the NLA gain, and $\chi=\tanh\left(\sinh^{-1}\left(\sqrt{\mu_{\textrm{res}}}\right)\right)$~\cite{Ralph2011}.

The EOF, though a valid entanglement measure, is an upper bound on the distillable entanglement, whereas the RCI is a lower bound on the distillable entanglement, and hence more operationally relevant to entanglement distillation. We analyze the average RCI of the heralded conditional covariance matrices across the EC box. One key mathematical difference between the measures is that while the EOF is non-negative by definition, the RCI can take on negative values for separable states. We find that averaging over the dual homodyne outcomes is severely detrimental to the RCI and leaves the average RCI negative for nearly all choice of parameters. Post-selecting on the dual homodyne outcome over a small range of values around 0 (whose probability of occurrence is maximal among all possible outcomes of the measurement), we plot the average of the RCI of the heralded conditional covariance matrices in Fig.~\ref{Ralphgrci} for EC box with $N=1,2-$NLA. These curves are identical to the ones in Fig.~\ref{rcivsmukappa01}, but are attained with higher values of mean photon number $\mu.$ This is consistent with the fact that the teleportation of two TMSV heralds a new TMSV of a different mean photon number, and hence, the teleportation through the EC box converges to transmission of a different TMSV through the NLA assisted lossy channel.

\section{Conclusions\label{sec:Conclusions}}

To summarize, we studied continuous-variable entanglement distillation with quantum scissors-based NLA from a noisy two-mode squeezed state shared across a pure loss channel, as shown in Figs.~\ref{half channel} and \ref{HC with mod scissors}. We presented a calculation based on phase space characteristic
functions and the Husimi-$Q$ function to determine the non-Gaussian
state heralded by the quantum scissors and the associated heralding probability. The complexity of the calculation scales efficiently with the number of quantum scissors. Having determined the heralded non-Gaussian state, we evaluated its Gaussian RCI and numerically optimized it over the input mean photon number and the NLA gain, where the RCI is a lower bound on the distillable entanglement per copy of the shared state when many copies are available that is achievable using one-way LOCC.

We also applied the calculation to the proposal of \cite{Ralph2011}
that replaces transmission through a lossy channel with teleportation
over a NLA-error corrected lossy entangled resource state as shown
in Fig.~\ref{Ralph EC box with mod scissors}. Previous studies on
this scheme had determined the logarithmic
negativity and a Gaussian lower bound on the EOF of the states heralded by single quantum scissors.
We validated some of these findings with our method and extended the analysis to the case of two quantum scissors. Additionally, we calculated the Gaussian RCI for the scheme with one and
two quantum scissors.

Our main conclusions include: a) In CV entanglement distillation over a pure loss channel, using the quantum scissors, it is possible to herald entangled states whose RCI exceeds the direct transmission entanglement distillation capacity of the channel $C_{\rm direct}(\eta)$ of Eq.~(\ref{plob}). b) Increasing the
number of scissors amounts to higher Gaussian RCI of the heralded state. The increase in heralded Gaussian RCI comes at the expense of a significantly lower success probability. c) In some cases, a second quantum scissors can help herald a Gaussian RCI that exceeds $C_{\rm direct}(\eta)$, while a single quantum scissors could not help exceed the bound. d) In the NLA-CV error correction scheme of Fig.~\ref{Ralph EC box with mod scissors},
the Gaussian RCI heralded by the scheme, on average
(over the teleportation dual homodyne detection outcomes), does not
exceed $C_{\rm direct}(\eta)$. Yet, when post-selected over a narrow
window of the teleportation dual homodyne measurement outcomes around
zero, it can exceed the same. In this limit of a small
window of outcomes, the scheme converges to the scheme in Fig.~\ref{half channel}\textemdash a pure loss channel appended by quantum scissors-based NLA with an entangled state input, albeit with higher optimal input mean photon numbers.

\begin{figure}
	\begin{tabular}{c}
		\includegraphics[scale=0.45]{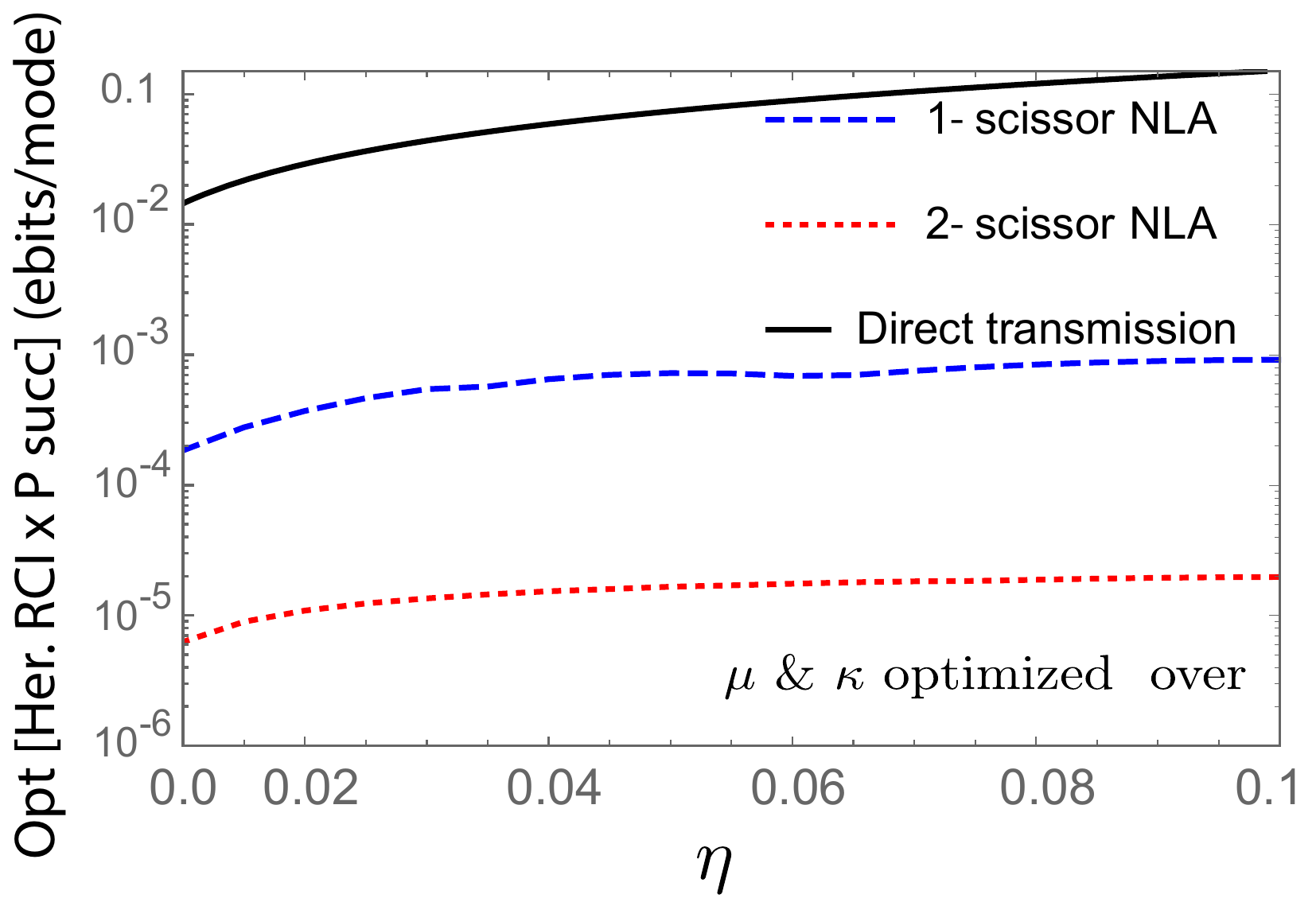}\tabularnewline
	\end{tabular}
	
	\centering\caption{Numerically optimized value of the product of heralded Gaussian RCI and the heralding success probability for $N=1,2-$NLA in the setup shown in Fig.~\ref{half channel}, as a function of the pure loss channel's transmissivity $\eta$. The value is optimized over the NLA gain and the mean photon number of the input TMSV state. The black curve corresponds to the direct transmission capacity $C_{\rm direct}(\eta)=-\log_2(1-\eta)$.}
	\label{optprod}
\end{figure}

Although our results show that the quantum scissors-NLA based scheme in Fig.~\ref{HC with mod scissors} can herald entangled states whose distillable entanglement exceeds $C_{\rm direct}(\eta)$, they do not demonstrate a quantum repeater. In order to demonstrate a quantum repeater using the quantum scissors, the product of heralded RCI and the heralding success probability, which is the true rate of entanglement distillation, must exceed $C_{\rm direct}(\eta)$. Clearly, as seen in Figs.~\ref{psuccvsmukappa01},~\ref{HCgaussRCIscatter},~\ref{Ralphgeofpsuccvsgsq} (b), the success probability of the quantum scissors drops steeply with increasing NLA gain so that the product of the heralded RCI and the heralding probability stays significantly below the direct transmission rate-loss tradeoff. This is further elucidated in Fig.~\ref{optprod}, where the product of heralded Gaussian RCI and the heralding success probability for the setup in Fig.~\ref{HC with mod scissors}, numerically optimized over the NLA gain parameter and the mean photon number of the input TMSV state, is plotted as a function of channel loss for number of quantum scissors $N=1,2$. The curves are found to be below $C_{\rm direct}(\eta)$, whereas the heralded RCI alone, e.g. at $\eta=0.01$, were seen to exceed $C_{\rm direct}(\eta)$ in Fig.~\ref{rcivsmukappa01}. Nevertheless, the fact that the quantum scissors are able to herald states with distillable entanglement higher than $C_{\rm direct}(\eta)$ paves the way towards constructing a multiplexing-based, second generation, CV quantum repeater scheme using the quantum scissors, as described in~\cite{SKG18}.

With regard to the experimental implementation of NLA using multiple quantum scissors, the primary imperfections to be considered include detection inefficiencies, single photon source inefficiencies and lack of photon number resolution. Our model for the quantum scissors in Fig.~\ref{HC with mod scissors} already addresses the latter two considerations, while the former remains to be analyzed. The heralded RCIs and the success probability of the quantum scissors are expected to drop when detection inefficiencies are considered.

\noindent\textbf{Note added-} During completion of this work, we
became aware of the work of~\cite{GORPR18}, which investigates the
use of quantum scissors-based NLA in the context of CV quantum key
distribution over a thermal loss channel using a method similar to
ours, and shows improved distance of transmission. They lower bound
the secret key generation rate in terms of the difference of a mutual
information and a Holevo information, whereas we have lower bounded
the entanglement and secret key distillation rates over a pure loss channel using the reverse
coherent information.

\begin{acknowledgments}
KPS thanks Animesh Datta, Christos Gagatsos, Stefano Pirandola, Mark M. Wilde and Zheshen Zhang for valuable discussions. This work was supported by the Office of Naval Research program Communications
and Networking with Quantum Operationally-Secure Technology for Maritime
Deployment (CONQUEST), awarded under Raytheon BBN Technologies prime
contract number N00014-16-C-2069, and a subcontract to University
of Arizona. This document does not contain technology or technical
data controlled under either the U.S. International Traffic in Arms
Regulations or the U.S. Export Administration Regulations.
\end{acknowledgments}

\bibliography{ref1}

\appendix
\begin{widetext}

\section{Mathematical Description of the System\label{sec:Gaussian-States,-Unitaries}}

\textbf{Gaussian States of a Bosonic Continuous-Variable (CV) system.}
A system of $M$ bosonic modes can be described by the creation and
annihilation operators $\hat{{a}}_{i}^{\dagger},\hat{{a}}_{i}$, such
that $\left[\hat{{a}}_{i},\hat{{a}}_{j}^{\dagger}\right]=\delta_{i,j},$
$\left[\hat{{a}}_{i},\hat{{a}}_{j}\right]=\left[\hat{{a}}_{i}^{\dagger},\hat{{a}}_{j}^{\dagger}\right]=0\ \forall\ i,j\in\left\{ 1,\ldots,M\right\} $,
and the corresponding quadrature operators $\hat{{x}}_{i}=\left(\hat{{a}}_{i}+\hat{{a}}_{j}^{\dagger}\right)/\sqrt{2},$
$\hat{{p}}_{i}=\left(\hat{{a}}_{i}-\hat{{a}}_{j}^{\dagger}\right)/\left(i\sqrt{2}\right)$,
such that $\left[\hat{{x}}_{i},\hat{{p}}_{j}\right]=i\delta_{i,j}.$

For a quantum state $\hat{{\rho}}$ defined on the $M-$mode Hilbert
space $\mathcal{H}^{\otimes M},$ a characteristic function can be
defined as the following operator Fourier transform (c.f., \cite{TJS15})
\begin{equation}
\chi\left(\xi\right)=\operatorname{Tr}\left(\hat{{\rho}}\hat{\mathcal{W}}\left(\xi\right)\right),\label{eq: char fn definition}
\end{equation}
where $\hat{\mathcal{W}}\left(\xi\right)$ is the Weyl operator 
\begin{equation}
\hat{\mathcal{W}}\left(\xi\right)=\exp\left(-i\xi^{T}\hat{\mathbf{r}}\right),\label{eq:Weyl op}
\end{equation}
and $\hat{\mathbf{r}}=\left(\hat{{x}}_{1},\ldots,\hat{{x}}_{M},\hat{{p}}_{1},\ldots,\hat{{p}}_{M}\right)^{T},$
$\xi=\left(\xi_{1},\ldots,\xi_{2M}\right)^{T},\ \xi_{i}\in\mathbb{R\ \forall}i\in\left\{ 1,\ldots,M\right\} .$
The characteristic function of (\ref{eq: char fn definition}) for
a quantum Gaussian state by definition is Gaussian, i.e., it can be
written as
\begin{equation}
\chi\left(\xi\right)=\exp\left(-\frac{{1}}{4}\xi^{T}\mathbf{V}\xi-i\mathbf{s}^{T}\xi\right),\label{eq: char fn Gauss}
\end{equation}
where $\mathbf{V}$ is the $2M\times2M$ real symmetric covariance
matrix defined as $\mathbf{V}_{i,j}=\left\langle \left\{ \hat{\mathbf{r}}_{i},\hat{\mathbf{r}}_{j}\right\} \right\rangle _{\rho}-2\left\langle \hat{\mathbf{r}}_{i}\right\rangle _{\rho}\left\langle \hat{\mathbf{r}}_{j}\right\rangle _{\rho}$
and $\mathbf{s}=\left\langle \hat{\mathbf{r}}\right\rangle _{\rho}$
is the $2M-$dimensional mean displacement vector.

The vacuum state is a Gaussian state with a covariance matrix equal
to the identity operator $I$. The two-mode squeezed vacuum (TMSV)
state of mean photon number $\mu=\sinh^{2}\left(r\right)$ ($r$ being
the squeezing parameter) is a Gaussian state with
\begin{equation}
\mathbf{V}^{\textrm{TMSV}}\left(\mu\right)=\left(\begin{array}{cc}
\mathbf{V}^{+}\left(\mu\right) & 0\\
0 & \mathbf{V}^{-}\left(\mu\right)
\end{array}\right),\;\mathbf{V}^{\pm}\left(\mu\right)=\left(\begin{array}{cc}
2\mu+1 & \pm2\sqrt{{\mu\left(\mu+1\right)}}\\
\pm2\sqrt{{\mu\left(\mu+1\right)}} & 2\mu+1
\end{array}\right),\label{eq:TMSV cov matrxi}
\end{equation}
and $\mathbf{s}=\vec{0}.$ The vacuum state is a special case of the
TMSV with $r=\mu=0.$

The covariance matrix of a quantum state satisfies the Heisenberg uncertainty principle $\mathbf{V}+i\Omega\geq 0$, where
\[
\Omega_n=\begin{pmatrix}0 & 1\\
-1 & 0
\end{pmatrix}\otimes I_{n\times n}.
\]
According to Williamson's theorem, a quantum covariance matrix $\mathbf{V}+i\Omega\geq 0$ can be diagonalized as 
\begin{equation}
\mathbf{V}=S_{V}\left(\mathbf{D}_{V}\oplus\mathbf{D}_{V}\right)S_{V}^{\operatorname{T}},
\end{equation}
where $S_{V}$ is a $2n\times2n$ Real symplectic matrix and $\mathbf{D}_{V}=\operatorname{diag}\left(\nu_{1},\ldots,\nu_{n}\right),$
where the $\nu_{j}$ are called the symplectic eigenvalues of $\mathbf{V}.$

\noindent\textbf{Gaussian Unitaries.} Unitary operators of the form
$\hat{U}_{\mathbf{s},S}=\exp\left(i\hat{H}\right),$ where $\hat{H}$
is a Hamiltonian that is at most quadratic in $\hat{\mathbf{r}}$
are called Gaussian unitaries. They map quantum Gaussian states into
quantum Gaussian states. An arbitrary Gaussian unitary operator can
be decomposed as 
\begin{equation}
\hat{U}_{\mathbf{s},S}=\hat{D}_{-\mathbf{s}}\hat{U}_{S},
\end{equation}
where $\mathbf{s}\in\mathbb{R}^{2M},$ $\hat{D}_{-\mathbf{s}}=\otimes_{j=1}^{M}\hat{D}_{-\left(\mathbf{s}_{j},\mathbf{s}_{M+j}\right)}$
is the displacement operator such that
\begin{equation}
\hat{D}_{-\left(\mathbf{s}_{j},\mathbf{s}_{M+j}\right)}=\exp\left(i\left(\mathbf{s}_{M+j}\hat{\mathbf{r}}_{j}-\mathbf{s}_{j}\hat{\mathbf{r}}_{+M+j}\right)\right),\label{eq:disp op}
\end{equation}
and $\hat{U}_{S}$ is a canonical Gaussian unitary operator generated
by a purely quadratic Hamiltonian.

A canonical Gaussian unitary operator $\hat{U}_{\mathbf{s},S}$ transforms
the quadrature operators as
\begin{equation}
\hat{\mathbf{r}}\rightarrow\hat{U}_{\mathbf{s},S}\hat{\mathbf{r}}\hat{U}_{\mathbf{s},S}^{\dagger}=S\hat{\mathbf{r}}+\mathbf{s},
\end{equation}
where $S$ is a $2M\times2M$ symplectic matrix and $\mathbf{s}\in\mathbb{R}^{2M}$.
Consequently, it transforms the first two statistical moments of an
arbitrary quantum state as
\begin{align}
\mathbf{s} & \rightarrow S\mathbf{s},\:V\rightarrow S\mathbf{V}S^{\operatorname{T}},
\end{align}
where $\mathbf{s}\in\mathbb{R}^{2M}$ is the mean vector and $\mathbf{V}$
the $2M\times2M$ covariance matrix. 

The two-mode beam splitter transformation is a canonical Gaussian
unitary transformation given by
\begin{equation}
\hat{U}_{BS}=\exp\left(i\theta\left(\hat{x}_{1}\hat{p}_{2}-\hat{p}_{1}\hat{x}_{2}\right)\right),
\end{equation}
where $t=\cos^{2}\theta\in\left[0,1\right]$ is the transmissivity
of the beamsplitter. The corresponding symplectic matrix is given
by%
\begin{equation}
S^{\left(t\right)}=\left(\begin{array}{cccc}
\sqrt{t} & \sqrt{1-t} & 0 & 0\\
-\sqrt{1-t} & \sqrt{t} & 0 & 0\\
0 & 0 & \sqrt{t} & \sqrt{1-t}\\
0 & 0 & -\sqrt{1-t} & \sqrt{t}
\end{array}\right).\label{eq:BS Symplectic}
\end{equation}

\noindent\textbf{Pure loss  channel.} The pure loss channel of transmissivity
$\eta$ is a Gaussian channel that maps Gaussian states to Gaussian states. It can be modeled as a beam splitter unitary transformation of the same transmissivity between the lossy mode and an environment
mode that is in the vacuum state. The action of the pure loss channel on the lossy mode is obtained by tracing out the environment mode, and can be expressed as
\begin{equation}
\mathcal{N}^{\left(\eta\right)}:\mathbf{V}\rightarrow X^{T}\mathbf{V}X+Y,\label{eq:pure loss channel math}
\end{equation}
where $X=\sqrt{\eta}I$ and $Y=\left(1-\eta\right)I.$ 

\noindent\textbf{Initial and pre-measurement state in Fig.~\ref{HC with mod scissors}.}
Since the pure loss channel is a Gaussian channel and the beam splitter
transformation is a Gaussian unitary operation, the scheme
depicted in Fig.~\ref{HC with mod scissors}, the quantum state across
the five modes initially, and prior to measurements in modes $A,\:B,\ C,\ Y,\ D$,
are both Gaussian state with zero displacement and covariance matrices
given by
\begin{align}
\mathbf{V}_{\textrm{initial}} & =\mathbf{V}_{AA'}^{\textrm{TMSV}}\left(\mu\right)\otimes\mathbf{V}_{C'D}^{\textrm{TMSV}}\left(\mu_{aux}\right)\otimes I_{B'},\\
\mathbf{V}_{\textrm{pre-meas}} & =S_{Y',C''}^{\left(1/2\right)}S_{B',C'}^{\left(\kappa\right)}\mathcal{N}_{A'\rightarrow Y'}^{\eta}\left(\mathbf{V}_{\textrm{initial}}\right)\left(S_{B',C'}^{\left(\kappa\right)}\right)^{T}\left(S_{Y',C''}^{\left(1/2\right)}\right)^{T},
\end{align}
respectively.

\section{Gaussian Measurements, Conditional Dynamics and CV Teleportation\label{sec:Gaussian-Conditional-Dynamics}}

A Gaussian measurement is a projection onto a quantum Gaussian state,
and thus is completely characterized by a mean vector and a covariance
matrix. 

\noindent\textbf{Homodyne and Heterodyne Detection.} Homodyne detection
on a single-mode, say of the $x-$quadrature, is the projection on
to the Gaussian state with mean vector and covariance matrix
\begin{align}
\mathbf{r}_{hom} & =\left(x_{hom},0\right)^{\operatorname{T}},\:\\
\mathbf{V}_{hom} & =\lim_{r\rightarrow\infty}\left(\begin{array}{cc}
\exp\left(-2r\right) & 0\\
0 & \exp\left(+2r\right)
\end{array}\right),
\end{align}
respectively, where $x_{hom}$ is measurement outcome and $r\in\mathbb{R}$
is the squeezing parameter. Heterodyne detection, likewise, is the
projection on to a coherent state with mean vector and covariance
matrix, respectively being,
\begin{align}
\mathbf{r}_{het} & =\left(x_{het},y_{het}\right)^{\operatorname{T}},\:\mathbf{V}_{het}=\left(\begin{array}{cc}
1 & 0\\
0 & 1
\end{array}\right),
\end{align}
where $x_{het}+iy_{het}\in\mathbb{C}$ is the measurement outcome.

\noindent\textbf{Dual Homodyne Detection.} Dual homodyne detection
is the continuous-variable analog of a Bell state measurement between
two modes $A$ and $B$. It is a projection of the two modes on to
a displaced EPR state (displaced infinitely squeezed TMSV state, which
is realized by mixing the two modes on a 50:50 beamsplitter, following
by orthogonal homodyne detections on the two modes ($\hat{x}$ measurement
on one mode and $\hat{p}$ measurement on the other). The mean vector
and covariance matrix of the measurement after the beam splitter transformation
of the two modes is given by 
\begin{align}
\mathbf{r}_{Dual-hom} & =\left(\gamma_{x},0,0,\gamma_{y}\right)^{\operatorname{T}},\\
\mathbf{V}_{Dual-hom} & =\lim_{r\rightarrow\infty}\left(\begin{array}{cc}
\exp\left(-2r\right) & 0\\
0 & \exp\left(-2r\right)
\end{array}\right)\oplus\left(\begin{array}{cc}
\exp\left(+2r\right) & 0\\
0 & \exp\left(-2r\right)
\end{array}\right),
\end{align}
where $\gamma_{x}+i\gamma_{y}\in\mathbb{C}$ is the measurement outcome.

\noindent\textbf{Gaussian Conditional Dynamics and overlap integrals.}
Consider a continuous-variable system of $n$ modes. Let $AB$ be
a bipartition of the modes such that subsystem $B$ consists of $m$
modes and subsystem $A$ consists of the remaining $n-m$ modes. Let
\begin{equation}
\mathbf{s}=\left(\begin{array}{c}
\mathbf{s}_{A}\\
\mathbf{s}_{B}
\end{array}\right),\:\mathbf{V}=\left(\begin{array}{cc}
\mathbf{V}_{A} & \mathbf{V}_{AB}\\
\mathbf{V}_{AB}^{\operatorname{T}} & \mathbf{V}_{B}
\end{array}\right)
\end{equation}
be the mean vector and covariance matrix of a quantum Gaussian state
$\hat{\rho}$ over the systems $A$ and $B.$ The quantum state obtained
in mode $A$ by tracing out subsystem $B,$ namely $\hat{\rho}_{A}=\operatorname{Tr}_{B}\left(\rho_{AB}\right)$
is also a quantum Gaussian state with mean vector and covariance matrix
given by
\begin{equation}
\mathbf{s}=\begin{array}{c}
\mathbf{s}_{A}\end{array},\:\mathbf{V}=\mathbf{V}_{A},
\end{equation}
respectively. On the other hand, when the subsystem $B$ is measured
by a Gaussian projective operator $\hat{\rho}^{G}$ of mean vector
$\mathbf{r}_{m}\in\mathbb{R}^{2m}$ and covariance matrix $\mathbf{V}_{m}$,
then the quantum state $\hat{\rho}_{A}$ conditioned on the measurement
outcome $\mathbf{r}_{m}\in\mathbb{R}^{2m}$ is a quantum Gaussian
state too, but its mean vector and covariance matrix are given by
\cite{GLS16,Serafini17} 
\begin{align}
\mathbf{s} & =\mathbf{s}_{A}+\mathbf{V}_{AB}\frac{1}{\mathbf{V}_{B}+\mathbf{V}_{m}}\left(\mathbf{r}_{m}-\mathbf{s}_{B}\right),\nonumber \\
\mathbf{V} & =\mathbf{V}_{A}-\mathbf{V}_{AB}\frac{1}{\mathbf{V}_{B}+\mathbf{V}_{m}}\mathbf{V}_{AB}^{\operatorname{T}},\label{eq:Gauss cond dynamics}
\end{align}
where the probability density function of the outcome $\mathbf{r}_{m}$
is given by the Gaussian overlap integral $p\left(\mathbf{r}_{m}\right)=\operatorname{Tr}\left(\hat{\rho}_{B}^{G}\hat{\rho}_{AB}\right)$,
which evaluates to %
{} 
\begin{equation}
p\left(\mathbf{r}_{m}\right)=\frac{\exp\left(-\left(\mathbf{r}_{m}-\mathbf{s}_{B}\right)^{\operatorname{T}}\frac{1}{\mathbf{V}_{B}+\mathbf{V}_{m}}\left(\mathbf{r}_{m}-\mathbf{s}_{B}\right)\right)}{\pi^{m}\sqrt{{\det\left(\mathbf{V}_{B}+\mathbf{V}_{m}\right)\}}}}.\label{eq:Gauss outcome prob density}
\end{equation}

\noindent\textbf{CV Teleportation of a TMSV across the EC Box of \cite{Ralph2011}.}
Consider the scheme in Fig.~\ref{Ralph EC box with mod scissors}.
Since the dual homodyne detection, the lossy channel, and the beamsplitters
in the quantum scissors are all Gaussian operations, the joint quantum
state across the modes prior to the measurements in the quantum scissors
is Gaussian. The mean and covariance matrix of this Gaussian can be
written down using (\ref{eq:TMSV cov matrxi}), (\ref{eq:BS Symplectic}),
(\ref{eq:pure loss channel math}) and (\ref{eq:Gauss cond dynamics}). 

Based on the observed dual homodyne outcome $\gamma,$ after the NLA
operation, a displacement correction unitary is applied on the modes
$A$ and $B,$ where these modes are displaced back by $g_{A}\left(-\gamma_{x},-\gamma_{y}\right)$
and $g_{B}\left(-\gamma_{x},+\gamma_{y}\right)$. Here $g_{A},\,g_{B}$
are classical gain parameters, which can be optimized over.

\section{Non-Gaussian Measurement based on ON-OFF Photodetection \& Gaussian
Overlap Integrals\label{sec:Non-Gaussian-Measurement}}

ON-OFF photodetection is a measurement scheme described by the positive
operator valued measure (POVM) elements
\begin{equation}
\Pi_{0}=\left|0\right\rangle \left\langle 0\right|,\:\Pi_{1}=I-\Pi_{0},
\end{equation}
where the projective measurement $\Pi_{0}$ is Gaussian, but $\Pi_{1}$
is not. In the modified quantum scissors operation considered in this
work, both in Figs.~\ref{HC with mod scissors}, \ref{Ralph EC box with mod scissors},
the heralding measurements of NLA are based on ON-OFF photodetection.

When the subsystem $B$ consisting of $m$ out of $n$ modes of a
CV system $AB$ in a quantum Gaussian state $\hat{\rho}_{AB}$ is
measured with OFF photodetection ($\Pi_{0}$ projection) on all the
$m$ modes, the conditional (Gaussian) quantum state on subsystem
$A$ and the probability of obtaining the OFF outcome across the $m$
modes follow from (\ref{eq:Gauss cond dynamics}) and (\ref{eq:Gauss outcome prob density}),
respectively, with $\mathbf{V}_{m}=I^{\otimes m}$ and $\mathbf{r}_{m}=0$.
The latter is the overlap integral $\operatorname{Tr}\left(\Pi_{0}\hat{\rho}\right),$
and simplifies to
\begin{equation}
p_{\bar{0}}=\operatorname{Tr}\left(\left(\Pi_{0}^{\otimes m}\right)_{B}\hat{\rho}_{AB}\right)=\frac{2^{m}\exp\left(-\mathbf{s}_{B}^{\operatorname{T}}\frac{1}{\mathbf{V}_{B}+I^{\otimes m}}\mathbf{s}_{B}\right)}{\sqrt{{\det\left(\mathbf{V}_{B}+I^{\otimes m}\right)}}},
\end{equation}
where $\mathbf{s}_{B}$ and $\mathbf{V}_{B}$are the mean and covariance
matrix of the modes in $B.$ 

Likewise, the probability of observing $\Pi_{1}$ is all the $m$
modes is given by
\begin{align}
p_{\bar{1}} & =\operatorname{Tr}\left(\left(\Pi_{1}^{\otimes m}\right)_{B}\hat{\rho}_{AB}\right)\\
 & =\operatorname{Tr}\left(\left(\Pi_{1}^{\otimes m}\right)_{B}\hat{\rho}_{B}\right)\\
 & =\operatorname{Tr}\left(\left(\mathbb{I}-\Pi_{0}\right)_{B}^{\otimes m}\hat{\rho}_{B}\right)\\
 & =\sum_{\tau\in\mathcal{P\left(K\right)}}\left(-1\right)^{\left|\tau\right|}\frac{2^{\left|\tau\right|}\exp\left(-\mathbf{s}_{\tau}^{\operatorname{T}}\frac{1}{\mathbf{V}_{\tau}+I_{\left|\tau\right|}}\mathbf{s}_{\tau}\right)}{\sqrt{{\det\left(\mathbf{V}_{\tau}+I_{\left|\tau\right|}\right)}}},
\end{align}
where $\mathcal{K}$ is the set of all $m$ modes contained in system
$B,$ $\mathcal{P\left(K\right)}$ the powerset of $\mathcal{K},$
i.e., the set of all subsets of $\mathcal{K}$ (inclusive of the null
element), $\mathbf{s}_{\tau}$ and $\mathbf{V}_{\tau}$ are the mean
vector and covariance matrix of the reduced quantum state on the modes
in element $\tau\in\mathcal{P\left(K\right)}$ and $I_{\left|\tau\right|}$
is the identity matrix of dimension $\left|\tau\right|.$ Though,
in this case the post measurement state on subsystem $A$ is non-Gaussian,
and hence cannot be captured using (\ref{eq:Gauss cond dynamics})
anymore. Nevertheless, the Husimi $Q$ function of the non-Gaussian
state on the modes in subsystem $A$ can be written down, e.g., when
$A$ consists of two modes $\hat{a},\,\hat{b},$ as
\begin{align}
Q\left(\alpha,\beta\right) & =\frac{\operatorname{Tr}\left(\left(\left|\alpha\right\rangle \left\langle \alpha\right|_{a}\otimes\left|\beta\right\rangle \left\langle \beta\right|_{b}\right)_{A}\otimes\left(\Pi_{1}^{\otimes m}\right)_{B}\hat{\rho}_{AB}\right)}{\pi^{2}p_{\bar{1}}}\\
\Rightarrow Q\left(\alpha_{x},\beta_{x},\alpha_{y},\beta_{y}\right) & =\frac{\sum_{\tau\in\mathcal{P\left(K\right)}}\left(-1\right)^{\left|\tau\right|}\frac{2^{\left|\tau\right|+2}\exp\left(-\left(\mathbf{s}_{\tau\cup A}-\mathbf{r}_{\tau\cup A}\right)^{\operatorname{T}}\frac{1}{\mathbf{V}_{\tau\cup A}+I_{\left|\tau\right|+2}}\left(\mathbf{s}_{\tau\cup A}-\mathbf{r}_{\tau\cup A}\right)\right)}{\sqrt{{\det\left(\mathbf{V}_{\tau\cup A}+I_{\left|\tau\right|+2}\right)}}}}{4\pi^{2}p_{\bar{1}}},
\end{align}
where $\alpha=\left(\alpha_{x}+i\alpha_{y}\right)/\sqrt{2}$ (and
likewise $\beta$), and $\mathbf{r}_{\tau\cup A}$ is the zero vector
except for the entries corresponding to the modes in $A,$ which take
the values $\left(\alpha_{x},\beta_{x},\alpha_{y},\beta_{y}\right).$ 

The same approach can be used to the construct the $Q$ function that
is heralded when some of the modes in $B$ are projected onto $\Pi_{0},$
while some others are projected onto $\Pi_{1},$ which is how we construct
the $Q$ function heralded by the $N-$quantum scissors NLA operations.

\section{Entanglement of Formation\label{sec:Entanglement-of-Formation}}
\begin{defn}
\label{def:eof}The entanglement of formation (EOF) of a bipartite
state $\rho_{AB}$ is defined as \cite{BDSW96}
\begin{equation}
E_{F}\left(\rho_{AB}\right):=\inf\left\{ \sum_{k}\lambda_{k}E\left(\left|\Psi_{k}\right\rangle \right)\left|\rho_{AB}=\sum_{k}\lambda_{k}\left|\Psi_{k}\right\rangle \left\langle \Psi_{k}\right|\right.\right\} ,
\end{equation}
where $\left|\Psi_{k}\right\rangle $ are entangled pure states and
$E\left(\left|\Psi_{k}\right\rangle \right)$ is the entanglement
entropy of $\left|\Psi_{k}\right\rangle .$ 

It is the minimum amount of pure entanglement required to construct
the state $\rho_{AB}.$ The EoF is non-increasing under local operations
and classical communication (LOCC).
\end{defn}

\begin{defn}
\label{def:geof}The Gaussian entanglement of formation (GEOF) of
a bipartite state $\rho_{AB}$ of mean vector $d$ and $4n\times4n$
dimensional covariance matrix $V$ ($2n$ total modes) is defined
as \cite{WGKWC04}
\begin{equation}
E_{G}\left(\rho_{AB}\left(V,d\right)\right):=\inf_{\lambda}\left\{ \int\lambda\left(dV_{p},d\xi\right)E\left(\Psi_{AB}^{G}\left(V_{p},\xi\right)\right)\left|\rho_{AB}=\int\lambda\left(dV_{p},d\xi\right)\Psi_{AB}^{G}\left(V_{p},\xi\right)\right.\right\} ,
\end{equation}
where $\Psi_{AB}^{G}$ are entangled Gaussian pure states and $\lambda$
is a measure in probability space. For a $n\left|n\right.-$mode bipartite
state (total $2n$ modes), the GEoF is given by
\begin{align}
E_{G}\left(\rho_{AB}\left(V,d\right)\right) & =\sum_{k=1}^{n}H\left(r_{k}\right),\\
H\left(r\right) & =\cosh^{2}\left(r\right)\log_{2}\left(\cosh^{2}\left(r\right)\right)-\sinh^{2}\left(r\right)\log_{2}\left(\sinh^{2}\left(r\right)\right).
\end{align}
This is so because every $n\left|n\right.-$mode bipartite pure Gaussian
state is a tensor product of $n$ two mode squeezed states with squeezing
parameters $r_{k},$ $k\in\left\{ 1,\ldots,,n\right\} $ up to a local
GLOCC unitary operation, and the entanglement of a TMS state with
squeezing $r$ is $H\left(r\right)$ as above.

It is the minimum amount of pure Gaussian entanglement required to
construct the state $\rho_{AB}.$ The GEOF is non-increasing under
Gaussian local operations and classical communication (GLOCC).
\end{defn}

\begin{cor}
The Gaussian entanglement of formation is at least as large as the
entanglement of formation
\begin{equation}
E_{G}\left(\rho_{AB}\right)\geq E_{F}\left(\rho_{AB}\right).
\end{equation}
\end{cor}

\begin{proof}
This follows from Definitions \ref{def:geof} and \ref{def:eof}.
The former is an infimum over a restricted set of possible decompositions
of the state than the latter, and hence is equal or larger than the
latter.
\end{proof}
\begin{lem}
\label{lem: Gaussian-extremality-geof}(Gaussian extremality of EOF
and GEOF) Among the set of all quantum states with covariance matrix
$V,$ and arbitrary mean and other moments, the (Gaussian) entanglement
of formation is minimized by the Gaussian states whose covariance
matrix equals $V,$ i.e.,
\begin{equation}
E_{F/G}\left(\rho_{AB}\left(V\right)\right)\geq E_{F/G}\left(\rho_{AB}^{G}\left(V\right)\right).
\end{equation}
\end{lem}

\begin{prop}
For any two-mode Gaussian state $\rho_{AB}^{G},$ the Gaussian entanglement
of formation equals its entanglement of formation, i.e.,
\begin{equation}
E_{G}\left(\rho_{AB}^{G}\right)=E_{F}\left(\rho_{AB}^{G}\right).
\end{equation}
\end{prop}

\begin{proof}
See \cite{MM08}.
\end{proof}
\begin{figure}
\medskip{}

\begin{tabular}{c}
 \includegraphics[scale=0.3]{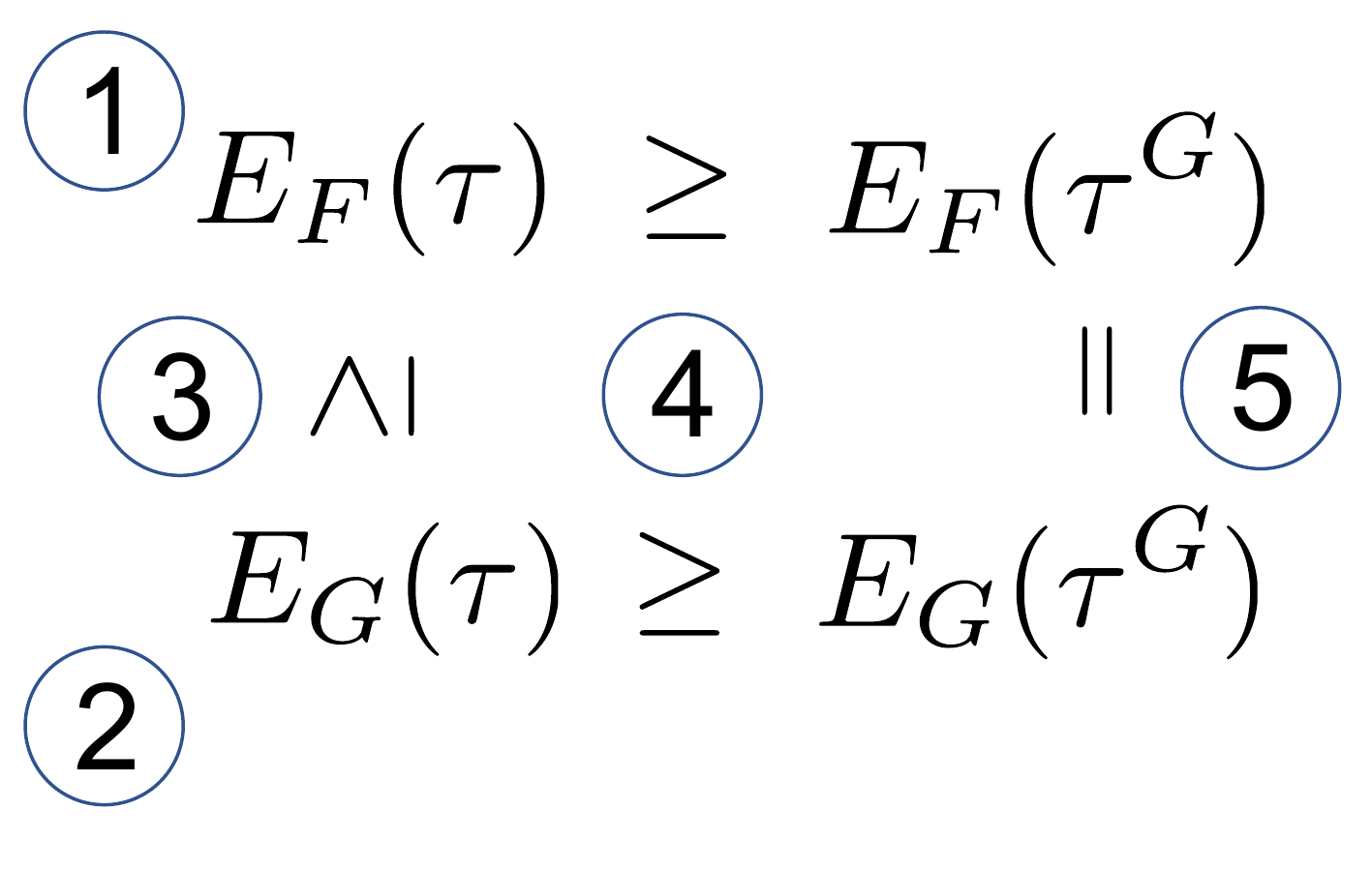}\tabularnewline
\end{tabular}

\centering\caption{Relation between entanglement of formation and Gaussian entanglement
of formation. $\tau$ represents a bipartite non-Gaussian state, while
$\tau^{G}$ denotes a Gaussian state whose covariance matrix is the
same as that of $\tau.$ The numberings correspond to the numberings
of the text in Appendix~\ref{sec:Entanglement-of-Formation}. }
\label{Ralph EC box with mod scissors-1}
\end{figure}
\begin{lem}
\label{lem:geof of gaussian state}The Gaussian entanglement of formation
of a bipartite Gaussian state of mean $d$ and covariance matrix $V$
equals
\begin{equation}
E_{G}\left(\rho_{AB}^{G}\left(V,d\right)\right):=\inf_{V_{p}}\left\{ E\left(\rho_{AB}^{G}\left(V_{p},0\right)\right)\left|V_{p}\leq V\right.\right\} ,
\end{equation}
where $\rho_{AB}^{G}\left(V_{p},0\right)$ are pure entangled states,
and $E$ is the entanglement entropy.
\end{lem}

\begin{proof}
See \cite{MM08,WGKWC04}. We use the results in \cite{MM08} to evaluate
the (Gaussian) entanglement of formation of the heralded covariance
matrix in the NLA-assisted communication schemes.
\end{proof}
\begin{rem}
\label{rem:geof of gaussian fn of cm}Evidently, from Lemma \ref{lem:geof of gaussian state},
the Gaussian entanglement of formation is independent of displacements
and equals the entanglement entropy of a TMSV state, wherein the infimum
picks the TMSV state with the smallest possible squeezing.
\end{rem}

\begin{cor}
\label{cor:ergodic geof}When the covariance matrix of a non-Gaussian
quantum state $\rho_{AB}$ is $V\left(\gamma\right)$, where $\gamma$
is some complex parameter distributed according to $P\left(\gamma\right)$,
we have that
\begin{equation}
\int d\gamma P\left(\gamma\right)E_{F/G}\left(\rho_{AB}\left(V\left(\gamma\right)\right)\right)\geq\int d\gamma P\left(\gamma\right)E_{F/G}\left(\rho_{AB}^{G}\left(V\left(\gamma\right)\right)\right).
\end{equation}
\end{cor}

\begin{proof}
This follows from Lemma \ref{lem: Gaussian-extremality-geof} and
the fact that $P\left(\gamma\right)\geq0$ and $E_{F/G}\geq0$ for
any state.
\end{proof}
\begin{rem}
We use the lower bound in Corollary \ref{cor:ergodic geof} (with
$E_{G}$) as our figure of merit for the scheme depicted in Fig.~\ref{Ralph EC box with mod scissors}.
\end{rem}

\begin{rem}
\label{rem:displ and means}A deterministic displacement operation
affects only the mean of a generic quantum state, and doesn't change
its covariance matrix or higher moments. 
\end{rem}

\begin{cor}
Given a generic conditional state $\rho_{AB}\left(\gamma\right)$
of mean $d\left(\gamma\right)$ and covariance matrix $V\left(\gamma\right),$
conditioned on a parameter $\gamma$ (e.g., $\gamma$ could be the
outcome of a dual homodyne detection), the action of conditional displacements
$D\left(g\gamma\right)$ on the state doesn't change its GEOF or the
average GEOF of Corollary (\ref{cor:ergodic geof}).
\end{cor}

\begin{rem}
Thus, teleportation displacement correction does not affect the ergodic
average GEOF lower bound we calculate for Ralph's scheme.
\end{rem}

\end{widetext}
\end{document}